\documentclass[10pt,oneside,twocolum,final]{IEEEtran}
\usepackage{setspace}
\usepackage{psfrag}
\usepackage{subfigure}
\usepackage{url}
\usepackage{stfloats}
\usepackage{amsthm}
\usepackage{amsmath}
\usepackage{amssymb}
\usepackage{amsfonts}
\usepackage[dvipdfm,unicode]{hyperref}
\usepackage{color}
\usepackage{bm}
\usepackage{graphicx}
\usepackage{placeins}
\usepackage{float}
\usepackage{tabularx,colortbl}
\usepackage{comment}
\usepackage{cite}
\usepackage{subfigure}
\begin{document}
\title{Spectral and Energy Efficiency of Multi-pair Massive MIMO Relay Network with Hybrid Processing}
\author{\IEEEauthorblockN{Wei~Xu,~\IEEEmembership{Senior~Member,~IEEE}, Jian~Liu,  Shi~Jin,~\IEEEmembership{Member,~IEEE}, and Xiaodai~Dong,~\IEEEmembership{Senior~Member,~IEEE}}\\
\thanks{Manuscript received May 2, 2016; revised September 27, 2016, January 21, 2017, and April 12, 2017; accepted June 6, 2017. This work of was supported in part by the 973 program under grant No. 2013CB329204, the NSFC under grant 61471114, the Six Talent Peaks project in Jiangsu Province under GDZB-005, and the Open Research Fund of the State Key Lab of ISN under ISN18-03. The work of S. Jin was supported in part by the NSFC for Distinguished Young Scholars of China with Grant 61625106. The work of X. Dong was supported in part by the NSERC of Canada under Grant 261524. The associate editor coordinating the review of this paper and approving it for publication was Dr. L. Liu. \emph{(Corresponding
author: W. Xu.)}}
\thanks{W. Xu, J. Liu, and S. Jin are with the National Mobile Communications Research Laboratory, Southeast University, Nanjing 210096, China (email:~\{wxu, liu\_jian, jinshi\}@seu.edu.cn). W. Xu is also with the State Key Lab of Integrated Services Networks, Xidian University, Xi'an, China.}
\thanks{X. Dong is with the Department of Electrical and Computer Engineering, University of Victoria, Canada (email: xdong@ece.uvic.ca).}
\thanks{Part of this work was presented at the WCNC 2016 in Doha, Qatar.}
\thanks{Color versions of one or more of the figures in this paper are available online
at http://ieeexplore.ieee.org.}
}

\maketitle

\newtheorem{mylemma}{Lemma}
\newtheorem{mytheorem}{Theorem}
\newtheorem{mypro}{Proposition}
\newtheorem{mycor}{Corollary}
\newtheorem{myrem}{Remark}

\begin{abstract}
We consider a multi-pair massive multiple-input multiple-output (MIMO) relay network, where the relay is equipped with a large number, $N$, of antennas, but driven by a far smaller number, $L$, of radio frequency (RF) chains. We assume that $K$ pairs of users are scheduled for simultaneous transmission, where $K$ satisfies $2K=L$. A hybrid signal processing scheme is presented for both uplink and downlink transmissions of the network. Analytical expressions of both spectral and energy efficiency are derived with respect to the RF chain number under imperfect channel estimation. It is revealed that, under the condition $N>\left\lfloor 4L^2/\pi \right\rfloor$, the transmit power of each user and the relay can be respectively scaled down by $1/\sqrt{N}$ and $2K/\sqrt{N}$ if pilot power scales with signal power, or they can be respectively scaled down by $1/N$ and $2K/N$ if the pilot power is kept fixed, while maintaining an asymptotically unchanged spectral efficiency (SE). While regarding energy efficiency (EE) of the network, the optimal EE is shown to be achieved when $P_r = 2K P_s$, where $P_r$ and $P_s$ respectively refer to the transmit power of the relay and each source terminal. We show that the network EE is a quasi-concave function with respect to the number of RF-chains which, therefore, admits a unique globally optimal choice of the RF-chain number. Numerical simulations are conducted to verify our observations.
\end{abstract}

\begin{IEEEkeywords}
Massive MIMO relay, hybrid processing, spectral efficiency (SE), energy efficiency (EE), limited RF chain.
\end{IEEEkeywords}

%
\IEEEpeerreviewmaketitle

\section{Introduction}
Massive multiple-input multiple-output (MIMO) has recently attracted much attention from both academia and industry \cite{larsson2014massive,rusek2013scaling,zhu2014secure}. When tens and hundreds of antennas are deployed at base station (BS), massive MIMO can provide significant performance enhancement \cite{marzetta2010noncooperative,ngo2013multicell,zhang2013large} in terms of both spectral efficiency (SE) and energy efficiency (EE). With the massive MIMO setup, random channel vectors between users and BS become asymptotically orthogonal to each other \cite{rusek2013scaling}. Moreover, both uncorrelated noise and the intra-cell interference disappear in the limit of an infinite number of antennas, even with simple matched filter at BS \cite{marzetta2010noncooperative}. Therefore, simple linear beamforming techniques, such as matched filtering, are capable of approaching the multiuser MIMO capacity expected by complicated non-linear precoding, namely dirty paper coding (DPC) \cite{ngo2013energy}.

Enhancing network performance with relays, e.g. in LTE Release 10 and beyond \cite{LTE-Book}, has emerged as an effective technique to expand cell coverage \cite{jin2010ergodic,gao2008channel,gao2009optimal}. Precoding designs and performance analysis of MIMO relay networks have been extensively investigated \cite{hammerstrom2007power,lee2010joint,xing2010robust}. Encouraged by the impressive merits of massive MIMO, relays equipped with massive antennas are naturally becoming attractive \cite{suraweera2013multi,ngo2014multipair,cui2014multi,jin2015ergodic}. Multi-pair one-way relaying with a large antenna array was studied in \cite{ngo2014multipair}, taking both ZF (zero-forcing) and MRC/MRT (maximum-ratio combining/maximum-ratio transmission) precoding into consideration. In \cite{cui2014multi,jin2015ergodic}, two-way relaying with massive antennas was analyzed with respect to the achievable rate, and power efficiency was also accordingly characterized under some typical scenarios.

Generally, these existing works focused on massive MIMO array with full radio frequency (RF) chains, where each antenna element is supported by a dedicated RF chain. Under the common MIMO implementation setup, precoding is entirely realized in the digital domain to suppress interference between data streams. However, the number of RF chains being exactly equal to the number of antennas becomes ``unacceptable'' in the massive MIMO setup. The tremendous RF chains bring prohibitively high energy consumption and hardware complexity \cite{heath2015overview,gao2015energy}. In order to address the problem, practical RF chain constraints has recently been considered. With only a limited number of RF chains, however, it would be hard to implement full dimensional digital precoding for massive antenna elements. Therefore, hybrid digital and analog precoding design is necessary for massive MIMO with limited RF chains \cite{liang2014low,ni2015near,ni2015hybrid,longbaole}.

Specifically for hybrid precoding, transmit signals are firstly precoded by a low dimensional digital precoding, and then analog (phase-only) precoding with a high dimension is enabled by using cost-effective analog phase shifters (APSs) \cite{liang2014low}. Under the practical RF-chain constraint, \cite{ni2015near} considered multi-stream transmission in point-to-point (P2P) massive MIMO and proposed a near-optimal matrix decomposition based hybrid precoding (MD-HP). Downlink transmission for multiuser massive MIMO with limited RF chains was studied in \cite{liang2014low} and \cite{ni2015hybrid}. The former assumed single-antenna users while the latter further considered the scenario with multi-antenna users.
\cite{longbaole} investigated the hybrid precoding design for multiuser MIMO downlink under frequency selective channels.
Note that the hybrid digital and analog processing is found particularly suitable for mmWave MIMO communications as it effectively reduces cost and power consumption of high-bandwidth mixed-signal devices \cite{Ayach2012capacity,Ayach2014Spatially}. In \cite{liang2014approach}, a beam steering scheme was proposed by aligning the beam for each user towards its strongest path at the analog domain. Asymptotic analysis showed that the proposed scheme was able to achieve the performance imposed by full-dimensional baseband ZF precoding.

The aforementioned literature concerning massive MIMO with hybrid beamforming mostly assumed fixed $L$ with $L\ll N$, while the restrictive condition between RF-chain number, $L$, and number of antennas, $N$, is unclear under the RF-chain constraint. Note that the relationship between the number of data streams, $N_s$, and the number of RF chains, $L$, has been well established in \cite{Sohrabi2016hybrid}. It was revealed that, if $L$ is equal or larger than $2N_s$, the hybrid beamforming scheme could achieve exactly the same performance as any full digital beamformer. In this paper, we explicitly give the relationship between the number of RF chains and the number of antennas, which is not a fundamental limit of the system but is originated from a technical condition later in the derivations and yields useful insights for multiuser interference cancelation. Moreover, energy efficiency (EE) of the system is rarely investigated in these studies. In this paper, we analyze both SE and EE of a multi-pair two-way relay network with a large antenna array at the relay while only limited RF chains are available for transmission and reception. We propose a hybrid precoding scheme for both uplink and downlink relaying. The main contributions of this paper are summarized as follows:
\begin{itemize}
\item An analytical expression is derived for characterizing the network SE in the presence of imperfect channel estimation. According to the derived result, power scaling laws are revealed for the RF-chain constrained system. Moreover, we explicitly derive the constraint between $L$ and $N$ as $N>\left\lfloor 4L^2/\pi \right\rfloor$. Even though this condition originates from a technical requirement in derivations, it yet could be interpreted as a sufficient condition to ensure some asymptotic system performance target. More specifically, the condition reveals that if the number of antennas $N$ is larger than $4L^2/\pi$, the multiuser interference can be effectively mitigated with the proposed hybrid processing scheme.
\item By considering the power consumption including transmit power, RF-chain power consumption as well as the power of massive APSs used for analog beamforming, we investigated the EE of the network. The optimal power allocation strategy for EE maximization is discovered as $P_r=2KP_s$, where $P_s$ and $P_r$ respectively represent the transmit power of a single source terminal and the relay.
\item Given the number of available RF-chains, there exists a unique globally optimal choice of transmit power, $P_s^*$, for maximizing the network EE. The relationship between the maximal EE and the corresponding SE is discovered. Especially under the case of perfect channel estimation, the maximal EE, $\eta_{EE}^*$, and its corresponding SE, $\eta_{SE}^*$, follows $\log(\eta_{EE}^*)=-\frac{\log2}{K}\eta_{SE}^* + c$ where $c$ is a constant depending on system parameters.
\end{itemize}

The rest of the paper is organized as follows. Section~\ref{sec:system_model} describes the system model of the RF-chain constrained massive MIMO relay network. Section~\ref{sec:rate} investigates the SE performance as well as power scaling law of the network. Section~\ref{sec:EE} characterizes the EE performance and presents insightful observations on tradeoffs between SE and EE. Simulation results are shown in Section~\ref{sec:simulation} before concluding remarks drawn in Section~\ref{sec:conclusion}.

\emph{Notations:} Throughout the paper, $\|\cdot\|_F$, $(\cdot)^T$, $(\cdot)^*$ and $(\cdot)^H$ represent the Frobenius norm, transpose, conjugation and Hermitian of a matrix, respectively. $\|\cdot\|$ denotes the Euclidean norm of a vector. $(\cdot)^{-1}$ returns the inverse of an invertible matrix. $\mathbb{E}\left\{\right\}$ and $\mathbb{V}\left\{\right\}$ take expectation and variance, respectively, while $\Pr()$ is the probability of an event. Operator $\text{arg}(\cdot)$ takes the angle of a complex number and $\text{diag}(\cdots)$ returns a (block) diagonal matrix with diagonal elements listed in the parentheses. Additionally, $\left\lfloor \cdot \right\rfloor$ means rounding down to the nearest integer.

\section{System Model}\label{sec:system_model}

We consider a multi-pair two-way relay network, where multiple pairs of single-antenna users exchange data within each pair via the help of an $N$-antenna relay. The relay is equipped with a massive antenna array. For massive MIMO, a large number of RF chains puts a heavy burden on the hardware cost and energy consumption. Besides, though the size of massive antenna array could be reduced by utilizing higher frequencies, the space occupied by the massive RF circuits may still prevent the device size to be made small enough. Therefore, an alternative way is to let the massive antenna array driven by a far smaller number of RF chains, say $L\ll N$, for the sake of efficient implementation. In this way, both hardware cost and power consumption can be significantly reduced. For the multiuser network, the number of overall users can be large and varying  which makes user scheduling necessary before transmission. We assume that $K$ user pairs are selected to be served simultaneously from the active user pool. Let the number of scheduled users be $2K=L$. We use $(k,k')$ to denote the user pair $k$ and $k'$ who exchange information with each other. For instance, the $j$-th communication pair is indexed as $(2j-1,2j)$, $j=1,2,\cdots,K$. Direct link between each pair $(k,k')$ is ignorable due to severe path loss and heavy shadowing.

\begin{figure}[htbp]
\centering
\subfigure[MAC phase: hybrid detection;]{
\begin{minipage}{0.8\linewidth}
\centering
\includegraphics[width=0.8\textwidth,height=0.3\textwidth]{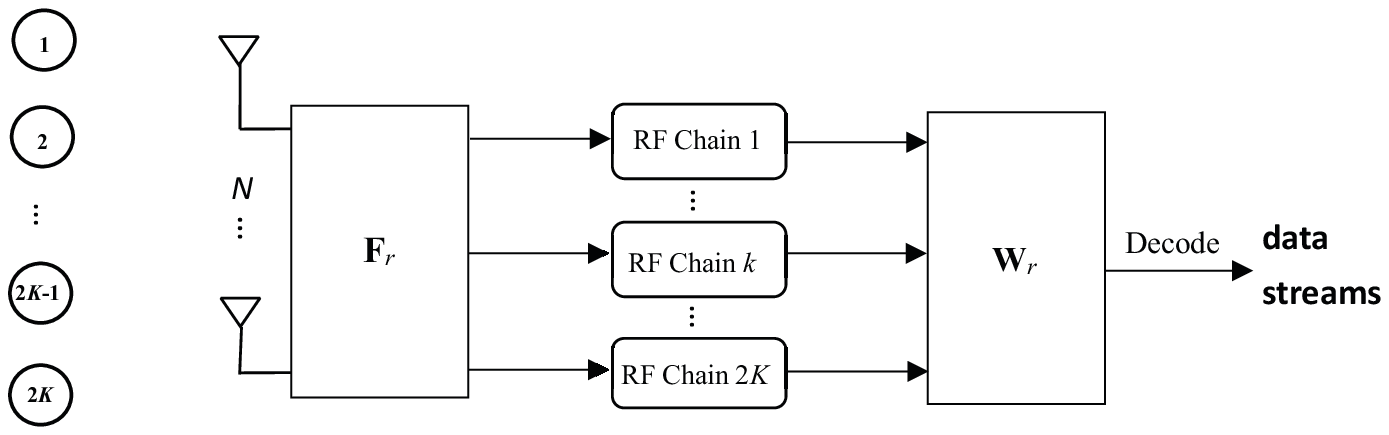}
\end{minipage}
}
\subfigure[BC phase: hybrid precoding.]{
\begin{minipage}{0.8\linewidth}
\centering
\includegraphics[width=0.8\textwidth,height=0.3\textwidth]{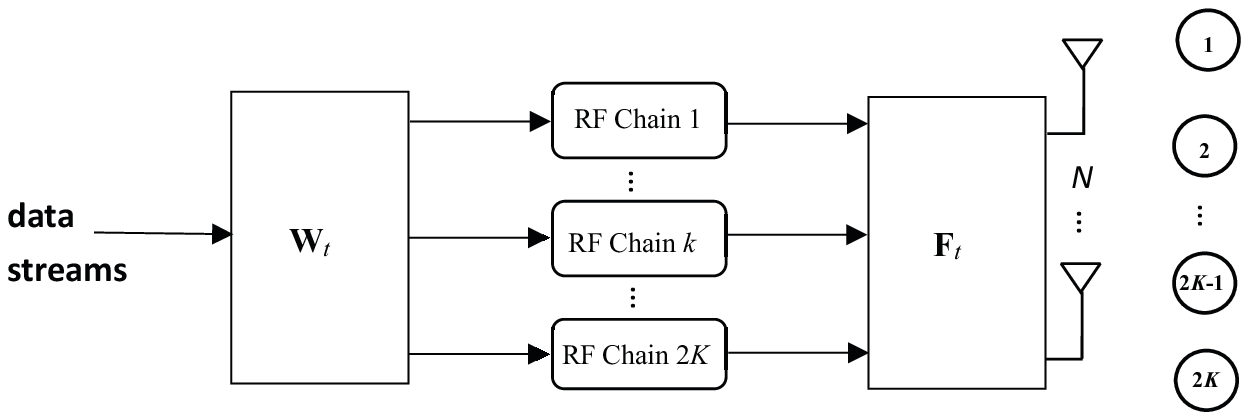}
\end{minipage}
}
\caption{RF-chain constrained multi-pair two-way relay network using hybrid precoding and detection.}
\label{Fig:System_Model}
\end{figure}

Assume that the relay and all users operate in time division duplex (TDD) mode and hence channel reciprocity can be ensured for the massive MIMO setup. All channels between the relay and users follow independent and identically distributed (i.i.d.) Rayleigh fading. The uplink channel from the $j$-th user to the relay is denoted as $\mathbf{g}_{j}\in\mathbb{C}^{N\times 1}$ whose entries are i.i.d. $\mathcal{CN}(0,\beta_j)$ where $\beta_j$ represents large scale fading. Then the uplink channel from all users to the relay is represented as $\mathbf{G}=[\mathbf{g}_{1}~\mathbf{g}_{2}\cdots\mathbf{g}_{2K-1}~\mathbf{g}_{2K}]$. Equivalently, the uplink channel matrix can be expressed as $\mathbf{G}=\mathbf{H}\mathbf{\Lambda}^{\frac{1}{2}}$ where $\mathbf{H}\in\mathbb{C}^{N\times2K}$ corresponds to small scale fading with all entries i.i.d. $\mathcal{CN}(0,1)$ and $\mathbf{\Lambda}\in\mathbb{C}^{2K\times2K}$ is a diagonal matrix representing large scale fading with $\mathbf{\Lambda}(j,j)=\beta_{j}$.

As illustrated in Fig. \ref{Fig:System_Model}, data transmission of the multi-pair two-way relay network is realized in two phases.

\noindent 1) Multiple access (MAC) phase: At time instant $t$, all users simultaneously transmit their independent signals to the relay. The received signal at the relay is given as
\begin{equation}\label{eq:relayreceive}
\mathbf{y}_R[t]=\sqrt{P_s}\mathbf{G}\mathbf{x}[t]+\mathbf{n}_R[t],
\end{equation}
where $\mathbf{x}=[x_1~x_2\cdots x_{2K-1}~x_{2K}]^T$ and $x_j$ is the transmit signal of the $j$-th user. It is assumed that all signals are normalized as $\mathbb{E}\{|x_j[t]|^2\}=1$ and each user takes transmit power $P_s$. $\mathbf{n}_R[t]\sim\mathcal{CN}(\mathbf{0},\mathbf{I}_N)$ is the additive white Gaussian noise (AWGN) at the relay.

\noindent 2) Broadcast (BC) phase: At time instant $t+1$, the relay broadcasts $\mathbf{s}[t+1]\in\mathbb{C}^{N\times1}$ to all users. Assume $\mathbb{E}\left\{\|\mathbf{s}[t+1]\|^2\right\}=1$ and denote $P_r$ as the average transmit power of relay. Applying channel reciprocity with calibration, the downlink channel from the relay to the $k'$-th user equals $\mathbf{g}_{k'}^T$.
Thus the received signal at the $k'$-th user is given as
\begin{align}
\mathbf{y}_{k'}[t+1]=&\sqrt{P_r}\mathbf{g}_{k'}^T\mathbf{s}[t+1]+n_{k'}[t+1],\label{eq:received}
\end{align}
where $n_{k'}[t+1]\sim\mathcal{CN}(0,1)$ is the AWGN at user $k'$.

\subsection{Channel Estimation}\label{sec:channel_est}

In a conventional MIMO relay system with full RF chains, channel estimation is conducted at the relay through orthogonal uplink pilots \cite{ngo2013energy}. However, in a hybrid structure, each RF chain is connected to multiple antennas. The traditional orthogonal estimation method could not be used any more since each RF chain receives the sum of signals from its coupled antennas, which is inseparable. So far as we know, efficient channel estimation is an interesting but still open problem for the hybrid system with limited RF chains. Most existing works on hybrid precoding design, like \cite{gao2015energy,ni2015near,liang2014low,ni2015hybrid,longbaole}, explicitly assumed perfect channel estimation. Though few have claimed that channel estimation with high accuracy can be obtained by exploiting the sparse property of mmWave channels \cite{Ayach2014Spatially,alkhateeb2014channel}, it should be pointed out that this kind of channel estimation has not been shown effective in general cases.

In this paper, we consider a simple round-robin estimation method. During each round, $2K=L$ relay antennas are chosen from the $N$ antennas and trained via the $L$ RF chains. To assist the heuristic channel estimation with limited RF chains, we may use a single fully-connected architecture by letting each phase shifter be controlled by an on-off switcher. Since we only have $L$ RF chains while $N$ channel coefficients are needed to be estimated from each user, we can estimate $L$ channel coefficients from the $N$ ones during each estimation phase and repeat for $\lceil N/L \rceil$ times in order to get all channel coefficients estimated. For data transmission, we simply let all the switches be ``on''. Since the corresponding channel coefficient is estimated when the corresponding switcher is ``on'' (the switchers connecting to other components that are not under estimation during this phase are ``off''), the obtained channel estimate is the same as the channel state of the same link when data transmission with the switchers definitely being ``on''. Though the method could be resource consuming and cost $N/L$ times more overhead, it is easy for implementation and provides a tractable approach for full-dimensional channel estimation with only a small number of RF chains. Even though some analytical results in our study could somewhat rely on the adopted channel estimation strategy, the performance characterization of the hybrid system could be further improved once a more efficient estimation method with low overhead emerges.

Denote $\hat{\mathbf{G}}$ as the estimate of $\mathbf{G}$ and $\bm{\mathcal{E}}_G=\mathbf{G}-\hat{\mathbf{G}}$ as the estimation error. Use $\hat{\mathbf{g}}_{j}$ and $\mathbf{\mathcal{E}}_{j}$ to represent the $j$-th column of $\hat{\mathbf{G}}$ and $\bm{\mathcal{E}}_G$, respectively. According to the property of MMSE estimation, we know that $\hat{\mathbf{G}}$ and $\bm{\mathcal{E}}_G$ are mutually independent \cite{estimationtheory} and
$$\hat{\mathbf{g}}_{j}\sim\mathcal{CN}(\mathbf{0},\sigma^2_{j}\mathbf{I}_N),~\mathbf{\mathcal{E}}_{j}\sim\mathcal{CN}(\mathbf{0},\varepsilon^2_{j}\mathbf{I}_N),$$
where
\begin{equation}\label{eq:channelestimation}
\sigma^2_{j}=\frac{\tau P_p\beta^2_{j}}{\tau P_p\beta_{j}+1},
~\varepsilon^2_{j}=\beta_{j}-\sigma^2_{j},
\end{equation}
where $P_p$ is the average power of pilot symbols and $\tau$ is the length of pilot sequences.

\subsection{Hybrid Processing at Relay}

In conventional massive MIMO systems, each antenna element is supported by a dedicate RF chain for digital signal processing. It has been shown in literature, like \cite{marzetta2010noncooperative,larsson2014massive,ngo2013energy}, that simple linear but full digital precoding schemes as ZF precoding is able to achieve asymptotically optimal performance in multiuser massive MIMO. While for a limited number of available RF chains, i.e., $L \le N$ as considered in this study, we have to exploit hybrid digital and analog processing techniques.

\subsubsection{Hybrid ZF Detection}

As shown in Fig.~\ref{Fig:System_Model}, the receiving matrix of the relay is made up of successive analog and digital processing components, denoted by $\mathbf{F}_{r}$ and $\mathbf{W}_{r}$, respectively. To specify, only phase rotations can be made through $\mathbf{F}_{r}$ while both amplitude and phase modifications are feasible by $\mathbf{W}_{r}$. Note that we exploit the hybrid processing approach presented in the previous \cite{liang2014low} which has shown to be asymptotically optimal with respect to SE in the one-hop massive MIMO. Accordingly, $\mathbf{F}_{r}$ is designed by extracting the phases of $\hat{\mathbf{G}}$:
\begin{equation}\label{eq:analogprecoder}
[\mathbf{F}_{r}]_{i,j}=\frac{1}{\sqrt{N}}e^{j\phi_{i,j}},
\end{equation}
where $\phi_{i,j}$ is the phase of the $(i,j)$-th element of $\hat{\mathbf{G}}^H$. It should be pointed out that for the case $2K<L$, the heuristic design in \eqref{eq:analogprecoder} could not be directly applicable, since dimensions of $\mathbf{F}_r$ and $\hat{\mathbf{G}}$ do not match. Therefore, for $2K<L$, we cannot design the analog precoding by direct channel phase extraction. A simple solution is to select $L'= 2K$ out of the $L$ RF chains, which
reduces to the trivial case. Otherwise, we could design the remaining $L-2K$ columns of $\mathbf{F}_r$ resorting to random phases as in \cite{Zhu2016Secure}, or through sophisticated optimization methods like in \cite{Yu2016alternating}, which in general are less likely to yield tractable expressions for performance analysis. Hence, in this paper, we focus on the case $2K=L$ and will show its asymptotic optimality under the considered scenario.

Since only channel estimate is available, the relay treats $\hat{\mathbf{G}}$ as the true channel. It considers $\mathbf{F}_{r}\hat{\mathbf{G}}\in\mathbb{C}^{2K\times2K}$ as the equivalent uplink channel seen from baseband and it generates the digital precoder as $\mathbf{W}_r=[\mathbf{F}_r\hat{\mathbf{G}}]^{-1}$ based on the popular ZF design. By applying hybrid ZF detection, the received signal $\mathbf{y}_R[t]$ is separated into $2K$ streams as follows:
\begin{equation}
\mathbf{r}[t]=\sqrt{P_s}\mathbf{W}_r\mathbf{F}_r\mathbf{G}\mathbf{x}[t]+\mathbf{W}_r\mathbf{F}_r\mathbf{n}_R[t].
\end{equation}
The $k$-th stream ($k$-th element) of $\mathbf{r}[t]$ is extracted for detecting $x_k[t]$ from the $k$-th user:
\begin{align}\label{eq:k-thelement}
r_{k}[t]&=\sqrt{P_s}\mathbf{w}_{k}^T\mathbf{F}_r\mathbf{g}_{k} x_{k}[t]+\sqrt{P_s}\sum_{j\neq k}^{2K}\mathbf{w}_{k}^T\mathbf{F}_r\mathbf{g}_{j} x_{j}[t]\nonumber\\
&+\mathbf{w}_{k}^T\mathbf{F}_r\mathbf{n}_R[t],
\end{align}
where $\mathbf{w}_{k}^T$ represents the $k$-th row of $\mathbf{W}_r$.

\subsubsection{Hybrid ZF Transmission}

After detecting signals from users, the relay multiplies them with a digital precoding matrix $\mathbf{W}_t\in\mathbb{C}^{2K\times 2K}$ and an analog precoder $\mathbf{F}_t\in\mathbb{C}^{N\times 2K}$ before broadcasting them to all users. Applying channel reciprocity and following the ZF precoding design, the downlink analog and digital precoders are respectively given as
\begin{equation}
\mathbf{F}_{t}={\mathbf{F}_r^T},~\mathbf{W}_t=\left[\hat{\mathbf{G}}^T\mathbf{F}_t\right]^{-1}\mathbf{P},
\end{equation}
 where $\mathbf{P}$ is a permutation matrix introduced to ensure that signal from the $k$-th user arrives at its corresponding pair $k'$. To be exact, $\mathbf{P}$ is a block diagonal matrix, defined as $\mathbf{P} =\text{diag}(\mathbf{P}_1,\cdots,\mathbf{P}_K),$
where $\mathbf{P}_i=
\begin{bmatrix}
0&1\\
1&0
\end{bmatrix}$.
It is directly verified that $\mathbf{W}_t={\mathbf{W}_r^T}\mathbf{P}.$

Assume that there exists a processing delay of $d$ symbols at the relay. The transmit signal of the relay is
\begin{equation}\label{eq:relaytransmit}
\mathbf{s}[t+1]=\mu\mathbf{F}_t\mathbf{W}_t\mathbf{x}[t-d],
\end{equation}
where $\mu$ is a normalization factor to satisfy average transmit power constraint, i.e., $\mathbb{E}\{\|\mathbf{s}[t+1]\|^2\}=1$. Then, according to \eqref{eq:received}, the received signal at the $k'$-th user equals
\begin{align}\label{eq:received1}
\mathbf{y}_{k'}[t+1]&=\mu\sqrt{P_r}\mathbf{g}_{k'}^T\mathbf{F}_t\mathbf{v}_{k}x_{k}[t-d]+n_{k'}[t+1]\nonumber\\
&+\sum_{j\neq k}^{2K}\mu\sqrt{P_r}\mathbf{g}_{k'}^T\mathbf{F}_t\mathbf{v}_{j}x_{j}[t-d],
\end{align}
where $\mathbf{v}_{k}$ is the $k$-th column of $\mathbf{W}_t$ and $n_{k'}[t+1]\sim\mathcal{CN}(0,1)$ is the AWGN at the $k'$-th user.

\subsection{Quantized Phase Shifters}
In \eqref{eq:analogprecoder}, it is assumed that ideal phase shifters are available which perfectly yield continuous phases without quantization. However, the implementation of such shifters is less feasible, or at least, too expensive due to hardware limitations. More realistic phase shifters are also discussed in a later part of this study. To be specific, quantized phase of each entry of $\mathbf{F}_r$ is chosen from the codebook $\bm{\Phi}$ based on the closest Euclidean distance.
\begin{equation}
\bm{\Phi}=\left\{0,\frac{2\pi}{2^B},\frac{2\cdot 2\pi}{2^B},\frac{3\cdot 2\pi}{2^B},\cdots,\frac{(2^B-1)\cdot 2\pi}{2^B}\right\},
\end{equation}
where $B$ denotes the number of quantization bits.

\section{Spectral Efficiency with Limited RF-chains}\label{sec:rate}
In this section, we analyze SE of the two-way massive MIMO relay network in terms of the achievable sum rate with limited RF chains for large $N$. Analytical expressions on power scaling laws are also presented.
\subsection{Analysis of Achievable Sum Rate}
Denote $R_{k\rightarrow k'}$ as the ergodic achievable rate for the transmission link $k\rightarrow R~(\text{Relay}) \rightarrow k'$. The sum rate of the network is expressed as
\begin{equation}\label{eq:total_rate}
R_{sum}=\sum\limits_{k\rightarrow k'} R_{k\rightarrow k'},
\end{equation}
where the $k$-th user and the $k'$-th user constitute a communication (user) pair. Recall that there are overall $2K$ transmission links.
Without loss of generality, we focus on the transmission link $k\rightarrow R \rightarrow k'$. Define $R_{k\rightarrow k'}$ as:
\begin{equation}
R_{k\rightarrow k'}=\frac{1}{2}\min\left\{R_{k\rightarrow R},R_{R\rightarrow k'}\right\},
\end{equation}
where $\frac{1}{2}$ exists because the transmission occupies two time slots and $\min\left\{\cdot,\cdot\right\}$ returns the minimum of two values. Let $R_{k\rightarrow R}$ and $R_{R\rightarrow k'}$ stand for the ergodic achievable rates of the two transmission links $k\rightarrow R$ and $R\rightarrow k'$, respectively. It follows:
\begin{align}
R_{k\rightarrow R}&=\mathbb{E}\left\{\log_2(1+\gamma_{k\rightarrow R})\right\},\nonumber\\
R_{R\rightarrow k'}&=\mathbb{E}\left\{\log_2(1+\gamma_{R\rightarrow k'})\right\},\nonumber
\end{align}
where, from \eqref{eq:k-thelement} and \eqref{eq:received1}, we have
\begin{align}
\gamma_{k\rightarrow R}=&\frac{P_s|\mathbf{w}_{k}^T\mathbf{F}_r\mathbf{g}_{k}|^2}{P_s\left(\sum_{j\neq k}^{2K}|\mathbf{w}_{k}^T\mathbf{F}_r\mathbf{g}_{j}|^2\right)+\|\mathbf{w}_{k}^T\mathbf{F}_r\|^2}, \nonumber\\
\gamma_{R\rightarrow k'}&=\frac{P_r|\mathbf{g}_{k'}^T\mathbf{F}_t\mathbf{v}_{k}|^2}{P_r\left(\sum_{j\neq k}^{2K}|\mathbf{g}_{k'}^T\mathbf{F}_t\mathbf{v}_{j}|^2\right)+\frac{1}{\mu^2}}.\nonumber
\end{align}

It is difficult to give precise closed-form expressions of $R_{k\rightarrow R}$ and $R_{R\rightarrow k'}$. However, as the relay is equipped with a massive antenna array, the received signals almost surely converge to their expectation according to the law of large numbers. Hence, we follow a popular methodology like in \cite{hassibi2003much}, and rewrite the received signal as the mean plus an additive uncorrelated ``effective'' noise term. It yields
\begin{align}
r_{k}[t]=\sqrt{P_s}\mathbb{E}\left\{\mathbf{w}_{k}^T\mathbf{F}_r\mathbf{g}_{k}\right\} x_{k}[t]+\tilde{n}_{k}[t],
\end{align}
where
\begin{align}
\tilde{n}_{k}[t]&=\sqrt{P_s}(\mathbf{w}_{k}^T\mathbf{F}_r\mathbf{g}_{k}-\mathbb{E}\left\{\mathbf{w}_{k}^T\mathbf{F}_r\mathbf{g}_{k}\right\}) x_{k}[t]\nonumber\\
&+\sqrt{P_s}\sum_{j\neq k}^{2K}\mathbf{w}_{k}^T\mathbf{F}_r\mathbf{g}_{j} x_{j}[t]+\mathbf{w}_{k}^T\mathbf{F}_r\mathbf{n}_R[t]
\end{align}
is considered as the effective noise. This methodology has been widely applied in massive MIMO due to the following considerations: 1) it yields a tractable rate expression, which is a lower bound of the rate; 2) it does not require instantaneous CSI at the receiver. Only statistical CSI is required. It is well-known from \cite{cover2012elements} that the worst-case uncorrelated additive noise is independent Gaussian with the same variance. By treating $\tilde{n}_{k}[t]$ as the worst-case noise, we obtain
\begin{equation}\label{eq:uplinkapprox}
\tilde{R}_{k\rightarrow R}=\log_2\left(1+\frac{P_s|\mathbb{E}\left\{\mathbf{w}_{k}^T\mathbf{F}_r\mathbf{g}_{k}\right\}|^2}{P_s\mathbb{V}\{\mathbf{w}_{k}^T\mathbf{F}_r\mathbf{g}_{k}\}+\text{MP}_{k}+\text{AN}_{k}}\right),
\end{equation}
where $\text{MP}_{k}$ and $\text{AN}_{k}$ refer to the multi-pair interference and additive noise effects, respectively, given by
\begin{align}
\text{MP}_{k}=P_s\sum_{j\neq k}^{2K}\mathbb{E}\left\{|\mathbf{w}_{k}^T\mathbf{F}_r\mathbf{g}_{j}|^2\right\}, ~~ \text{AN}_{k}=\mathbb{E}\{\|\mathbf{w}_{k}^T\mathbf{F}_r\|^2\}.
\end{align}

Following the similar procedures and using \eqref{eq:received1}, we can also obtain
\begin{align}\label{eq:received2}
\tilde{R}_{R\rightarrow 'k}=&\log_2\bigg(1+\nonumber\\
&\frac{P_r|\mathbb{E}\left\{\mathbf{g}_{k'}^T\mathbf{F}_t\mathbf{v}_{k}\right\}|^2}{P_r\left[\mathbb{V}\{\mathbf{g}_{k'}^T\mathbf{F}_t\mathbf{v}_{k}\}+\sum_{j\neq k}^{2K}\mathbb{E}\left\{|\mathbf{g}_{k'}^T\mathbf{F}_t\mathbf{v}_{j}|^2\right\}
\right]+\frac{1}{\mu^2}}\bigg).
\end{align}
Thus for the massive MIMO relay network with large $N$, the analytical expression for $R_{k\rightarrow k'}$ can be characterized as
\begin{equation}\label{eq:rateappro}
R_{k\rightarrow k'}=\frac{1}{2}\min\left\{\tilde{R}_{k\rightarrow R},\tilde{R}_{R\rightarrow k'}\right\}.
\end{equation}
Now we are ready to evaluate the sum rate with $2K = L \ll N$ in the following theorem.
\begin{mytheorem}\label{theo:sumrate}
With $N>\left\lfloor 4L^2/\pi \right\rfloor$ satisfied, the ergodic achievable rate for the two-way relay network with $L$ RF chains can be characterized as $R_{sum}=\sum\limits_{k\rightarrow k'}R_{k\rightarrow k'}$, where
\begin{align}\label{eq:limitedRF}
\lim\limits_{N\rightarrow \infty}\frac{R_{k\rightarrow k'}}{\frac{1}{2}
\log_2(1+\frac{\pi N}{4}x)}=1,
\end{align}
in which we denote
\begin{equation}\label{eq:x_definition}
x=\min \bigg\{ \frac{P_s \sigma_{k}^2}{1+P_s\sum_{j=1}^{2K}\varepsilon_{j}^2},\frac{P_r } {(1+P_r\varepsilon^2_{k'})\sum_{j=1}^{2K}\frac{1}{\sigma^2_{j}}} \bigg\}
\end{equation}
for notational simplicity.
\end{mytheorem}
\begin{proof}
This theorem directly follows by characterizing the rates in \eqref{eq:total_rate} and \eqref{eq:rateappro} by applying \emph{Lemma}~\ref{lemma_uplink} and \emph{Lemma}~\ref{lemma_downlink}, respectively, in Appendices A and B.
\end{proof}
\begin{myrem}
Theorem 1 characterizes the sum rate of the considered RF-chain constrained relay network in the limit of an infinite number of antennas. Yet it could also serve as a good approximation for finite but large $N$. In particular, there are two conditions for $N$ so that \eqref{eq:limitedRF} could be a good approximation. Firstly, since the proof for Lemma~\ref{lemma_uplink} involves the central limit theorem (CLT), $N$ should be large enough to make the CLT sufficiently accurate. Secondly, $N$ should also satisfy the condition $N>\left\lfloor 4L^2/\pi \right\rfloor$, in order to make matrix expansion \eqref{eq:inversion} converge in probability. For fixed $L$, the condition $N>\left\lfloor 4L^2/\pi \right\rfloor$ always asymptotically holds as $N\rightarrow\infty$. However, if $L$ and $N$ are both large values, not only should $N$ be large enough to satisfy the CLT, but also $N$ should be larger than $\left\lfloor 4L^2/\pi \right\rfloor$ in order to ensure multiuser interference (MUI) to be effectively mitigated. We should note that the condition originates from a technical perspective. We do not necessarily need to design the system rigorously obeying this rule, but it can be viewed as a sufficient condition to guarantee that MUI can be effectively mitigated.
\end{myrem}

It should be pointed out that the analog precoder in \eqref{eq:analogprecoder}, originated for one-hop case \cite{liang2014low}, was shown to be asymptotically optimal. The following proposition justifies its asymptotic optimality in the considered relay scenario by examining the asymptotic behavior of $R_{k\rightarrow k'}$.

\begin{mypro}
 For the considered two-way relay network with massive antennas and limited RF chains, the ergodic per-user rate $R_{k\rightarrow k'}$ satisfies
\begin{align}
\lim\limits_{N\rightarrow\infty}\frac{R_{k\rightarrow k'}}{R_{k\rightarrow k'}^{\text{full}}}=1,
\end{align}
where $R_{k\rightarrow k'}$ is given by \eqref{eq:limitedRF}, while
\begin{equation}
R_{k\rightarrow k'}^{\text{full}}=\frac{1}{2}\log_2 \left(1+(N-2K)x\right)
\end{equation}
is the performance of the full digital ZF precoding in \cite[Eq. (45)]{ngo2014multipair}.
\end{mypro}
\begin{proof}
It is checked that
$$\lim\limits_{a\rightarrow\infty}\frac{\log_2(1+a)}{\log_2a}=\lim\limits_{a\rightarrow\infty}\frac{\log_2a+\log_2(1+1/a)}{\log_2a}=1.$$
Therefore, we have
\begin{align}
&\lim\limits_{N\rightarrow\infty}\frac{R_{k\rightarrow k'}}{R_{k\rightarrow k'}^{\text{full}}}\nonumber\\
=&\lim\limits_{N\rightarrow\infty}\frac{\log_2(1+\pi N x/4)}{\log_2\left(1+(N-2K)x\right)}\nonumber\\
=&\lim\limits_{N\rightarrow\infty}\frac{\log_2(1+\pi N x/4)}{\log_2\left(1+N x\right)}\times\lim\limits_{N\rightarrow\infty}\frac{\log_2\left(1+N x\right)}{\log_2\left(1+(N-2K)x\right)}\nonumber\\
=&\lim\limits_{N\rightarrow\infty}\frac{\log_2(\pi N x/4)}{\log_2\left(N x\right)}\times\lim\limits_{N\rightarrow\infty}\frac{\log_2\left(N x\right)}{\log_2\left((N-2K)x\right)}\nonumber\\
=&\lim\limits_{N\rightarrow\infty}\frac{\log_2(N x)+\log_2(\pi/4)}{\log_2(Nx)+\log_2(1-2K/N)}\nonumber\\
=&1.
\end{align}
\end{proof}

\subsection{Power Scaling Law}\label{sec:power_scaling}

We look into the potential of power saving of the system when the relay is equipped with a massive antenna array. In order to make expressions concise, the following assumes all channels between users and the relay experience equal path loss, i.e., $\beta_j=\beta_0$ for $j=1,2,\cdots,2K$. Then \eqref{eq:limitedRF} reduces to
{\small
\begin{align}\label{eq:equalpathloss}
R_{sum} = K \log_2 \bigg(1+\min\bigg\{\frac{P_s\pi N\sigma_0^2}{4(1+2KP_s\varepsilon_0^2)},
\frac{P_r\pi N \sigma_0^2} {8K(1+P_r\varepsilon_0^2)}\bigg\}\bigg),
\end{align}
}
where $\sigma_j=\sigma_0$ and $\varepsilon_j=\varepsilon_0~(j=1,2,\cdots,2K)$.

\noindent 1) Case 1: Assume that $P_p$ is a constant, which corresponds to the case where channel estimation accuracy remains unchanged. We try to find the potential for power saving in the transmission phase. Without loss of generality, we let $P_s$ and $P_r$ be scaled down proportionally by the factor of $1 / N^{\alpha}~(\alpha>0)$, i.e., $P_s=E_s/N^{\alpha}$ and $P_r=2K E_r/N^{\alpha}$, where $E_s$ and $E_r$ are fixed power budgets regardless of $N$. Now we elaborate in the following that any choice of $0< \alpha \leq 1$ is able to maintain a nonvanishing sum rate when $N\rightarrow\infty$.

Accordingly, by scaling down both source and relay power as stated above, it directly yields
\begin{align}\label{eq:case1_up}
\frac{P_s\pi N\sigma_0^2}{4(1+2KP_s\varepsilon_0^2)}=\frac{E_s\pi\sigma_0^2}{4(N^{\alpha-1}+2KE_s\varepsilon_0^2/N)},
\end{align}
and
\begin{equation}\label{eq:case1_down}
\frac{P_r\pi N \sigma_0^2} {8K(1+P_r\varepsilon_0^2)}=\frac{E_r\pi\sigma_0^2} {4(N^{\alpha-1}+2K\varepsilon_0^2E_r/N)}.
\end{equation}
Substituting \eqref{eq:case1_up} and \eqref{eq:case1_down} into \eqref{eq:equalpathloss}, we have
\begin{align}\label{eq:scaling_alpha}
R_{sum} = K \log_2 \bigg(1 + & \min\bigg\{\frac{E_s\pi\sigma_0^2}{4(N^{\alpha-1}+2K\varepsilon_0^2E_s/N)},\nonumber\\
&\frac{E_r\pi\sigma_0^2} {4(N^{\alpha-1}+2K\varepsilon_0^2E_r/N)}\bigg\}\bigg).
\end{align}
Now from \eqref{eq:scaling_alpha}, in order to maintain a nonvanishing sum rate performance when $N\rightarrow\infty$, the above expression implies that it is necessary to guarantee $\alpha-1\leq0$, which, recalling $\alpha>0$, yields $0<\alpha\leq 1$.
Furthermore, it can be inferred from \eqref{eq:scaling_alpha} that the sum rate tends to infinity for $0<\alpha<1$ when $N\rightarrow\infty$ even though the transmit power is scaled down by a factor of $1/N^{\alpha}$. In particular, for the special case with $\alpha=1$, the sum rate converges to a constant as
\begin{equation}
R_{sum} \rightarrow K \log_2 \bigg(1+\min\bigg\{\frac{E_s \pi\sigma_0^2}{4},
\frac{E_r\pi\sigma_0^2}{4} \bigg\}\bigg),
\end{equation}
when $N$ grows unboundedly. It implies that the transmit power of each user can be scaled down by $1/N$ and the relay transmit power can be cut down by a factor of $2K/N$ while maintaining the same performance for increasing $N$.

\noindent 2) Case 2: Apart from scaling down $P_s$ and $P_r$, we also consider the potential for saving pilot power $P_p$. Assume $P_p=P_s=E_s/N^{\alpha}$ and $P_r=2KE_r/N^{\alpha}$. To keep a nonvanishing sum rate in this scenario, we cannot just scale down the transmit power as aggressively as in Case 1, e.g., by $1/N$.

Let us substitute $P_p=E_s/N^{\alpha}$ into \eqref{eq:channelestimation} and consider $\sigma_j=\sigma_0$ and $\varepsilon_j=\varepsilon_0~(j=1,2,\cdots,2K)$. It is checked for $N\rightarrow\infty$ that
\begin{align}
N^{\alpha}\sigma_0^2&=\frac{\tau\beta_0^2E_s}{\tau\beta_0E_s/N^{\alpha}+1}\rightarrow\tau\beta_0^2E_s, \nonumber\\
\varepsilon_0^2&=\frac{\beta_0}{\tau\beta_0E_s/N^{\alpha}+1}\rightarrow\beta_0.
\end{align}
Further substituting the above results into \eqref{eq:equalpathloss}, we have the following equalities
\begin{align}
\frac{P_s\pi N\sigma_0^2}{4(1+2KP_s\varepsilon_0^2)}=\frac{\pi\tau E_s^2 \beta_0^2}{4(N^{2\alpha-1}+2K\beta_0 E_s N^{\alpha-1})},
\end{align}
and
\begin{equation}
\frac{P_r\pi N \sigma_0^2} {8K(1+P_r\varepsilon_0^2)}=\frac{\pi\tau E_s E_r \beta_0^2}{4(N^{2\alpha-1}+2K\beta_0 E_r N^{\alpha-1})}.
\end{equation}
Consequently, the sum rate with the scaled power becomes
\begin{align}\label{eq:powerscaling2}
R_{sum}=K\log_2\bigg(1 + &\min\bigg\{\frac{\pi\tau E_s^2 \beta_0^2}{4(N^{2\alpha-1}+2K\beta_0 E_s N^{\alpha-1})},\nonumber\\
&\frac{\pi\tau E_s E_r \beta_0^2}{4(N^{2\alpha-1}+2K\beta_0 E_r N^{\alpha-1})}\bigg\}\bigg).
\end{align}
When $N$ tends to infinity, it is critical to have $2\alpha-1\leq0$ and $\alpha-1\leq0$ in order to guarantee a nonvanishing sum rate, which yields $0<\alpha\leq\frac{1}{2}$.

Therefore, if $P_p$ is kept equal to $P_s$, each user and the relay can respectively be scaled down by factors of $1/ \sqrt{N}$ and $1/(\sqrt{N}/2K)$ while guaranteeing an asymptotically unchanged rate. This can be explained as, cutting the transmit power of pilots decreases the channel estimation accuracy, hence degrading the system performance further.

\section{Network Energy Efficiency Analysis}\label{sec:EE}

In this section, we investigate the energy efficiency of the relay network. As addressed in Section \ref{sec:channel_est}, effective channel estimation under a hybrid structure remains an open problem. Considering channel estimation is not the focus of our paper, in this section, we assume that channel estimate has been obtained using the simple scheme proposed in Section \ref{sec:channel_est}, and then we focus on the EE of the data transmission period. EE is defined as the ratio of SE over the total power consumption. Let $\eta_{SE}$ and $\eta_{EE}$ denote the SE and network EE, respectively. It follows
\begin{equation}\label{eq:EEdef}
\eta_{EE}=\frac{\eta_{SE}}{P_{sum}}=\frac{R_{sum}}{P_{sum}},
\end{equation}
where $R_{sum}$ is from \eqref{eq:limitedRF} and $P_{sum}$ refers to the total power consumption.
Before characterizing $\eta_{EE}$, it is necessary to introduce a proper power consumption model for this RF chain constrained relay network. Here we adopt a general but simple power consumption model \cite{han2015large} which is helpful in revealing useful observations. Accordingly, the total power consumption is written as:
\begin{equation}\label{eq:Psum}
P_{sum}=\frac{1}{2}(2KP_s/\kappa_U+P_r/\kappa_{r})+2KP_0+P_{const}+2KNP_{APS},
\end{equation}
where $\kappa_U<1$ and $\kappa_r<1$ represent the efficiency of power amplifiers (PAs) deployed at terminals and the relay, respectively, the terms in parentheses refer to the total transmit power consumption of the network, and the factor $\frac{1}{2}$ exists due to the fact of half duplexing. Term $2 K P_0$ is the power consumption that scales with the RF chains, $P_{const}$ represents the part of constant circuit power regardless of the RF chain number, and $2KNP_{APS}$ is the power consumption of all APSs used for analog beamforming. Note that a switching network is needed in the training period between the RF chains and antennas, so that each time $L$ antennas are chosen to be trained. Since the switching network mainly works for the channel estimation period as discussed above in our system setup, we here focus on EE for data transmission period and do not specifically taken the power consumption of the switching network into account.
For notational simplicity, let $\kappa_r=\kappa_U=\kappa$ and assume $\beta_j=\beta_0$ for $j=1,2, \cdots, 2K$. Then, from \eqref{eq:equalpathloss}, \eqref{eq:EEdef} and \eqref{eq:Psum}, it gives
\begin{align}\label{eq:EE1}
\eta_{EE}=\frac{K\log_2\bigg(1+\frac{\pi N\sigma_0^2}{8K}\Delta\bigg)}{\frac{1}{2\kappa}(2KP_s+P_r)+2KP_0+P_{const}+2KNP_{APS}},
\end{align}
where we define $\Delta=\min\bigg\{\frac{2KP_s}{1+2KP_s\varepsilon_0^2},
\frac{P_r} {1+P_r\varepsilon_0^2}\bigg\}$.

\subsection{Optimal Transmit Power for EE Maximization}

From \eqref{eq:EE1}, the expression of EE looks involved depending on a number of system parameters as well as circuit design. While given a fixed number of RF chains, it is interesting to conclude in the following \emph{Theorem} that the network EE can always be optimized via a proper power splitting strategy irrespective of the circuit depending parameters like $P_0, P_{const}$ and $P_{APS}$ individually.

\begin{mytheorem}\label{optimal_power}
For fixed RF-chain number and relay antenna number, there exists a globally optimal choice with respect to $(P_s^*,P_r^*)$ satisfying $2KP_s^*=P_r^*$ which yields the maximal EE, $\eta_{EE}^*$. The relationship between $\eta_{EE}^*$ and $P_s^*$ is given as
\begin{align}\label{eq:EE_star}
\eta_{EE}^*=\frac{2 a_0 \kappa}{[(a_0 + a_1)P_s^*+4](a_1 P_s^*+4)\ln2},
\end{align}
where $a_0 = \pi N\sigma_0^2$ and $a_1 = 8K\varepsilon_0^2$ respectively relate to the strength of received signal and interference.
\end{mytheorem}
\begin{proof}
The optimal relationship of $P_s^*$ and $P_r^*$ has been revealed by \emph{Lemma}~\ref{lemma_power} in Appendix~\ref{proof_lemma_power}. Accordingly, the maximal EE is always achieved when $P_r=2KP_s$. Substituting this condition into \eqref{eq:EE1}, the EE optimization problem is then formulated as
\begin{align}\label{eq:EE2}
\underset{P_s}{\text{maximize}} \quad &\eta_{EE}(P_s)=\frac{K\log_2\bigg(1+\frac{a_0 P_s}{a_1 P_s+4}\bigg)} {\frac{2KP_s}{\kappa}+P_{c}}\\
\text{subject to:} \quad & P_s>0,\nonumber
\end{align}
where $P_{c}=2KP_0+P_{const}+2KNP_{APS}$ stands for power consumption excluding the transmit power.
From \emph{Proposition}~\ref{Ps_star} in Appendix~\ref{optimal_Ps}, we show that $\eta_{EE}$ in \eqref{eq:EE2} is quasi-concave with respect to $P_s$. Therefore, from \cite{boydconvex}, there exists a globally optimal $P_s^*$ for EE maximization. By applying the KKT conditions of the optimization problem in \eqref{eq:EE2} and after some basic calculations, we have
\begin{equation}\label{eq:KKT_Ps}
\frac{2 a_0 \kappa \left( \frac{2KP_s^*}{\kappa}+P_{c} \right)}{[(a_0 + a_1)P_s^*+4](a_1 P_s^*+4)}-K\ln\bigg(1+\frac{a_0 P_s^*}{a_1 P_s^*+4}\bigg)=0.
\end{equation}
By plugging \eqref{eq:KKT_Ps} into the objective of \eqref{eq:EE2}, it gives the optimal EE as
\begin{equation}\label{green_point_relationship}
\eta_{EE}^*=\frac{2 a_0 \kappa}{[(a_0 + a_1)P_s^*+4](a_1 P_s^*+4)\ln2},
\end{equation}
which completes the proof.
\end{proof}
The specific value of $P_s^*$ for \eqref{eq:EE2} depends on the individual circuit parameters $(P_0,P_{const},P_{APS})$. However, it is worth noting that the corresponding $\eta_{EE}^*$ becomes irrespective of individual circuit parameters, once the value of $P_s^*$ has been determined through the set of circuit parameters $(P_0,P_{const},P_{APS})$, as expressed by \eqref{green_point_relationship}. More specifically, once $P_s^*$ is calculated from \eqref{eq:EE2} based on specific circuit parameters, the corresponding $\eta_{EE}^*$ is determined by \eqref{green_point_relationship} as a function of $P_s^*$ only. From \eqref{eq:EE_star}, it implies that a higher $P_s^*$ always leads to a lower $\eta_{EE}^*$. Further since the SE performance increases monotonically with $P_s$, a higher $\eta_{SE}^*$ always corresponds to a lower $\eta_{EE}^*$.

\begin{mycor}\label{cora}
For the case of perfect CSI, we have $\eta_{EE}^*=c_0 4^{-\frac{\eta_{SE}^*}{2K}}$ where  $c_0=\frac{a_0 \kappa}{8\ln2}$.
\end{mycor}
\begin{proof}
For perfect CSI, we have $\varepsilon_0^2=0$ which yields $a_1=0$. Then \eqref{eq:EE_star} reduces to
\begin{equation}\label{eq:perfect_simplify}
\eta_{EE}^*=\frac{2\kappa a_0}{4 \ln 2 (a_0 P_s^*+4)}.
\end{equation}
Further by substituting $\eta_{SE}^*=K\log_2(1+\frac{a_0 P_s^*}{4})$ into \eqref{eq:perfect_simplify}, it directly gives the corollary.
\end{proof}

Recalling the relationship $2K=L$, Corollary \ref{cora} reveals that the maximal EE $\eta_{EE}^*$ scales as exponentially decaying with respect to $\frac{1}{2K} \eta_{SE}^*$, which can be regarded as the corresponding SE w.r.t. $P_s^*$ of a single RF-chain. Due to the nature of exponentially decaying, there exists a sensitive region with small $\eta_{SE}^*$, where a slight increase in $\eta_{SE}^*$ would greatly decrease $\eta_{EE}^*$; While in the other region with high $\eta_{SE}^*$, $\eta_{EE}^*$ decreases slowly with an increasing $\eta_{SE}^*$. Therefore, whether the green point $\eta_{EE}^*$ falls into the sensitive region is of great importance when considering EE-SE tradeoff. Equivalently, if we take the logarithm on both sides of $\eta_{EE}^*=c_0 4^{-\frac{\eta_{SE}^*}{2K}}$, it directly gives
\begin{equation}\label{eq:green_point_line}
\log(\eta_{EE}^*)=-\frac{\log2}{K}\eta_{SE}^*+c,
\end{equation}
where
\begin{equation}\label{eq:Fig8_c}
c = \log c_0=\log\left(\frac{a_0 \kappa}{8\ln2}\right).
\end{equation} The above relationship indicates that the green points $(\log(\eta_{EE}^*),\eta_{SE}^*)$ actually lie in a straight line with a slope of $-\frac{\log2}{K}$. It implies that a smaller RF chain number admits more EE improvement under a given SE reduction with the optimal EE design. Note that this observation will later be verified by numerical results.

\subsection{Effect of RF-chain Number on EE}

In this part, we look into the effects of limited RF chains on the EE performance. For notational simplicity, we rewrite $\eta_{EE}$ in \eqref{eq:EE2} as
\begin{align}\label{eq:EE_K}
\eta_{EE}=\frac{K\log_2\big(1+\frac{b}{K+a}\big)}{d(K+m)},
\end{align}
where $a=\frac{1}{2P_s\varepsilon_0^2}, b=\frac{\pi N\sigma_0^2}{8\varepsilon_0^2}, d=2P_s/\kappa+2P_0+2NP_{APS}$ and $m=\frac{P_{const}}{d}$.
\begin{mypro}\label{Pro:K_concave}
Given a large but fixed number of relay antennas, there exists a globally optimal choice of the RF chain number, $L^*$, for EE maximization.
\end{mypro}
\begin{proof}
Since $2K=L$ is set in the system, the problem is equivalent to find the optimal $K$ for EE maximization. Denote $\eta_{EE}=\frac{f(K)}{g(K)}$ from \eqref{eq:EE_K}, where $f(K)=K\log_2\bigg(1+\frac{b}{K+a}\bigg)$ and $g(K)=d(K+m)$.
After some mathematical manipulations, the second-order derivative of $f(K)$ is given as
\begin{align}
f''(K)=- \frac{b[a(2K+2a+b)+b(K+a)]}{(K+a+b)^2(K+a)^2\ln2}<0,
\end{align}
which indicates that $f(K)$ is concave with respect to $K$.
Meanwhile, it is direct to see that $g(K)$ is linear, thus convex, with regard to $K$. Both $f(K)$ and $g(K)$ are differentiable. Hence, $\eta_{EE}(K)$ is a quasi-concave function with regard to $K$, which from \cite{boydconvex} implies that there exists a globally optimal $K^*$ for EE maximization.
\end{proof}
%

Note that the numerator of $\eta_{EE}$ in \eqref{eq:EEdef} is the SE expression which is accurate when $N>\left\lfloor 4L^2/\pi \right\rfloor$ is guaranteed. Because the denominator of $\eta_{EE}$ scales linearly with $L$, we infer that the maximal EE cannot be achieved at large $L$. Under this circumstance, there exists an optimal $L$ which falls into $N>\left\lfloor 4L^2/\pi \right\rfloor$ and makes our EE analysis accurate. This statement is also verified via simulations. As a growing $N$ up to $1000$ antennas yielding $L<\sqrt{\frac{\pi N}{4}}\approx 28$ RF chains, we will later justify that the EE achieves its maximum with $L<28$ agreeing with numerical verifications.

\section{Numerical Results}\label{sec:simulation}

\subsection{Spectral Efficiency}

\emph{1) Large Rayleigh Fading Channels}

In this section, the obtained observations are verified through Monte-Carlo simulations. In the simulations, we define $\mathrm{SNR}=P_s$ and set $P_r=2KP_s$ without otherwise being specifically stated. We set $K=5$ and training period is $\tau=2K$. The path loss effect is normalized to $\beta_j=1$ for $j=1,2,\cdots,2K$.

\begin{figure}[htbp]
\centering
\includegraphics[width=0.4\textwidth, height=0.34\textwidth]{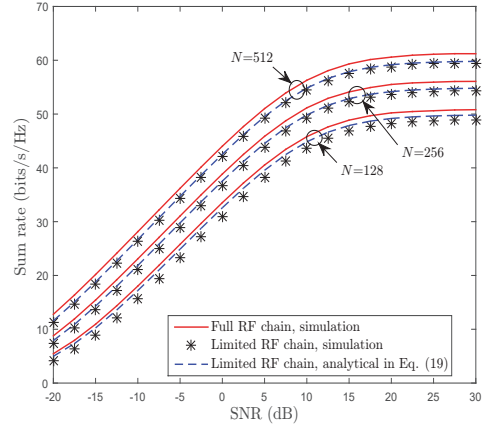}
\caption{Sum rate versus SNR with $P_p=10$dB.}
\label{Fig:sumrate_SNR_Pp}
\end{figure}

\begin{figure}[htbp]
\centering
\includegraphics[width=0.4\textwidth, height=0.34\textwidth]{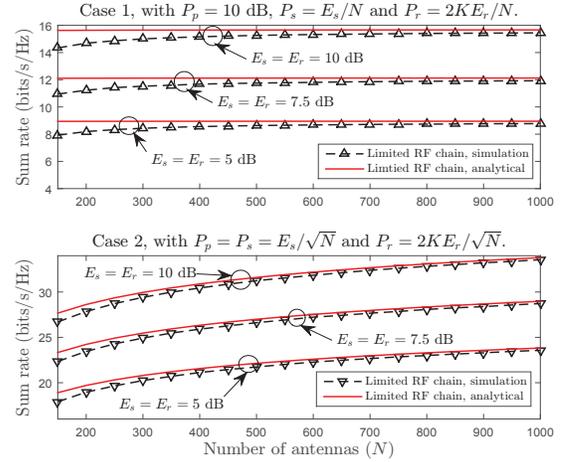}
\caption{Sum rate versus number of antennas for power scaling Cases 1 and 2.}
\label{Fig:powerscaling}
\end{figure}

Fig.~\ref{Fig:sumrate_SNR_Pp} shows the sum rate of the relay network versus system SNR under $P_p = 10$ dB. For comparison, the sum rate of the full RF-chain case is also provided as a benchmark. It is validated that our derived analytical result matches well with the exact sum rate and it becomes more accurate as $N$ grows larger. It reveals that the hybrid ZF scheme performs measurably close to the full RF-chain case. Notice that the sum rate performance is shown to saturate in the large SNR regime due to the effect of imperfect channel estimate. Moreover, in Fig.~\ref{Fig:powerscaling}, the effects of power scaling for Cases 1 and 2 are exemplified, respectively. As plotted in the figure, the sum rate converges for both cases when $N\rightarrow \infty$, as predicted in Section~\ref{sec:power_scaling}. However, the sum rate for Case 1 converges much faster than Case 2 as $\sqrt{N}$ grows at a much slower speed compared to $N$. Scaling down the pilot power in Case 2 introduces a further degradation of the system performance due to increased channel estimation error. Therefore, the data transmission power cannot be cut down with the same scaling law as in Case 1 in order to maintain the asymptotically same performance.

In Fig.~\ref{Fig:quantized}, the impact of quantized phase shifters on the proposed hybrid scheme is presented under various values of $B$. It is observed that there exists a significant gap between $B=1$ and ideal phase shifters, while $B=2$ performs quite close to the ideal case. For a further increased number of quantization bits $B=4$, the performance almost perfectly agrees with the ideal phase shifters. Therefore, it can be concluded that the proposed hybrid precoding scheme could give satisfying performance even with heavily quantized phase shifters, like $B=2$ and $B=4$.
\begin{figure}[htbp]
\centering
\includegraphics[width=0.4\textwidth, height=0.35\textwidth]{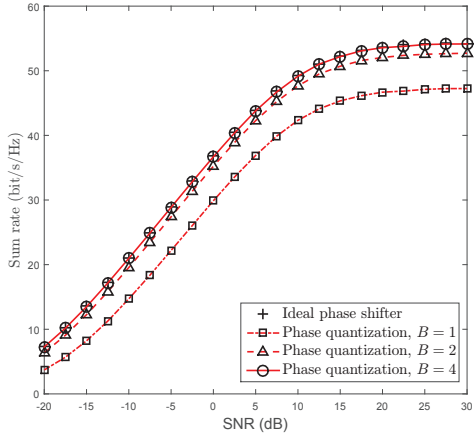}
\caption{Sum rate versus SNR under different phase shifter quantization bits, with $2K=10$, $N=256$ and $P_p=10$dB.}
\label{Fig:quantized}
\end{figure}
\begin{figure}[htbp]
\centering
\includegraphics[width=0.4\textwidth, height=0.35\textwidth]{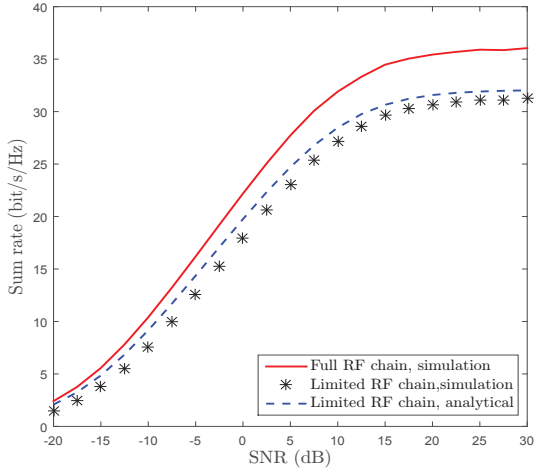}
\caption{System throughput versus SNR, with $N=64$, $\tau=2K=8$, $T=600$ and $P_p=10$dB.}
\label{Fig:Overhead}
\end{figure}

\emph{2) Possible Impacts of Channel Estimation Overhead}

By incorporating the training overhead caused by the proposed channel estimation, Fig. \ref{Fig:Overhead} presents the system throughput versus SNR. Specifically, the sum rate is multiplied by a factor $\eta_{limited}=\frac{T-L\cdot N/L}{T}=\frac{T-N}{T}$. For comparison, the sum rate for the full digital processing is multiplied by $\eta_{full}\frac{T-\tau}{T}=\frac{T-2K}{T}$. It is calculated that $\frac{\eta_{full}}{\eta_{limited}}=\frac{T-2K}{T-N}\approx 1.1$ under a typical massive MIMO setup $2K=4$, $N=64$ and $T=600$, which makes no significant change to the throughput, as shown in Fig. \ref{Fig:Overhead}.

Moreover, the impact of channel coherence time $T$ is depicted in Fig. \ref{Fig:coherence_time}. The observations show that the performance gap between full RF chain and limited RF chain decreases as $T$ becomes larger, which does indicate the importance of an effective channel estimation strategy for the limited RF chain system especially for high mobility scenarios with small $T$.

\begin{figure}[htbp]
\centering
\includegraphics[width=0.4\textwidth, height=0.35\textwidth]{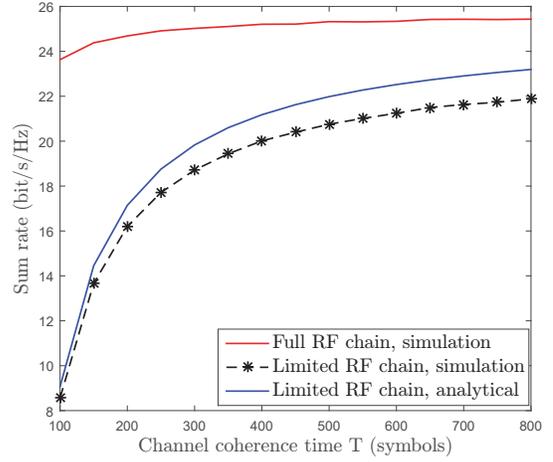}
\caption{System throughput versus channel coherence time $T$, with $\tau=2K=8$, $N=64$, $P_s=P_p=5$dB and $P_r=2KP_s$.}
\label{Fig:coherence_time}
\end{figure}

\begin{figure}[htbp]
\centering
\includegraphics[width=0.4\textwidth, height=0.35\textwidth]{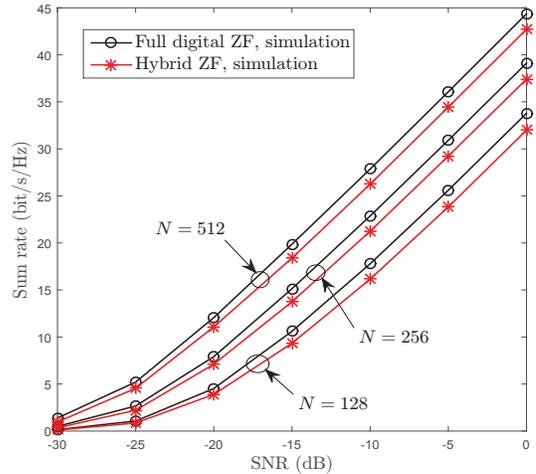}
\caption{Sum rate versus SNR under mmWave channel scenarios, with $K=5$, $d=\frac{1}{2}$ and $N_p=10$.}
\label{Fig:mmWave_N}
\end{figure}

\emph{3) Large mmWave Channels}

Apart from i.i.d. Rayleigh fading channels, we also investigate the performance of our proposed hybrid scheme under the geometric model, which is currently leveraging as a more realistic channel model for mmWave massive MIMO.
Specifically, the downlink channel from user $k$ to the relay, denoted as $\mathbf{g}_{k}^{T}$, can be characterized as \cite{ngo2013multicell,alkhateeb2014channel,Ayach2014Spatially}
\begin{equation}
\mathbf{g}_{k}^{T}=\sqrt{\frac{N}{N_p}}\sum\limits_{l=1}^{N_p}\alpha_{l}^{k}\bm{\alpha}^H(\theta_{k,l}),
\end{equation}
where each user is assumed to observe the same number of propogation paths, denoted by $N_p$, $\alpha_{l}^{k}$ is the gain of the $l$-th path of user $k$ distributed as $\mathcal{CN}(0,1)$, and $\theta_{k,l}$ is the random azimuth angle of departure drawn independently from uniform distributions over $[0,2\pi]$. $\bm{\alpha}^H(\theta_{k,l})$ is the array response vector depending on array structures. If we assume a uniform linear array (ULA) here, it can be given as
\begin{equation}
\bm{\alpha}^H(\theta_{k,l})=\frac{1}{\sqrt{N}}\left[1,e^{j2\pi d\sin(\theta_{k,l})},\cdots, e^{j(N-1)2\pi d\sin(\theta_{k,l})}\right],
\end{equation}
where $d$ refers to the normalized antenna spacing.
It is observed from Fig. \ref{Fig:mmWave_N} that under the more practical channel setup, the proposed hybrid processing scheme still performs remarkably close to the full digital precoding, which indicates that the advantages of the proposed scheme persist for more realistic large mmWave channels.

\begin{figure}[htbp]
\centering
\includegraphics[width=0.4\textwidth, height=0.35\textwidth]{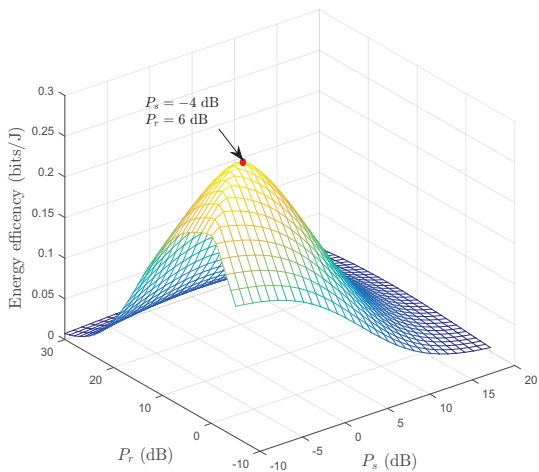}
\caption{Three-dimensional plot of EE versus $P_s$ and $P_r$, with $N=256$, $\kappa=0.375$, $P_p=10$dB, $P_{const}=20$W, $P_0=1$W and $P_{APS}=0.02$W.}
\label{Fig:3Dplot}
\end{figure}
\begin{figure}[htbp]
\centering
\includegraphics[width=0.4\textwidth, height=0.35\textwidth]{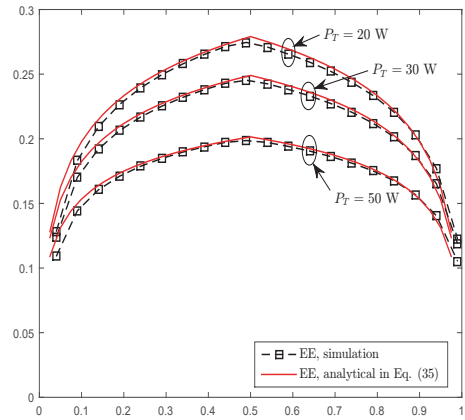}
\caption{EE under total power constraint $2KP_s+P_r=P_T$, with $N=256$, $\kappa=0.375$, $P_p=10$dB, $P_{const}=20$W, $P_0=1$W and $P_{APS}=0.02$W.}
\label{Fig:power_constraint}
\end{figure}

\subsection{Energy Efficiency}

Fig.~\ref{Fig:3Dplot} presents a 3D plot of EE with respect to $P_s$ and $P_r$. By searching the maximal EE, it is found out that the optimal EE in this test is achieved at $P_s = -4$ dB and $P_r = 6$ dB satisfying $P_r = 2 K P_s$ with $K=5$. This coincides with our observation in \emph{Theorem}~2 that the maximal EE is always achieved at $P_r=2KP_s$. Alternatively, Fig.~\ref{Fig:power_constraint} depicts the EE performance under the total transmit power constraint $2KP_s+P_r=P_T$ for different values of $P_T$. The horizontal axis represents the ratio between $P_r$ and the total transmit power. It can be inferred from the figure that the maximal EE is always achieved at $P_r=\frac{1}{2}P_T$, i.e. $P_r=2KP_s$, for any transmit power constraint $2KP_s+P_r=P_T$, which also verifies the first part of \emph{Theorem}~2.

To further illustrate the condition of $P_r=2KP_s$ for EE maximization, the contour plot of EE is presented in Fig.~\ref{Fig:Contour}. For comparison, the contour of total power constraint $2KP_s+P_r=P_T$ is plotted in black dotted lines while the contour for EE performance is plotted in red solid lines. It is found that the contours for EE are tangent to the contours for $P_T$ and the tangent points lie exactly in the line labeled as $P_r=2KP_s$ in the figure. This convinces us that the power allocation in terms of EE maximization is to set $2KP_s=P_r$ under an arbitrary total power constraint.
\begin{figure}[htbp]
\centering
\includegraphics[width=0.4\textwidth, height=0.35\textwidth]{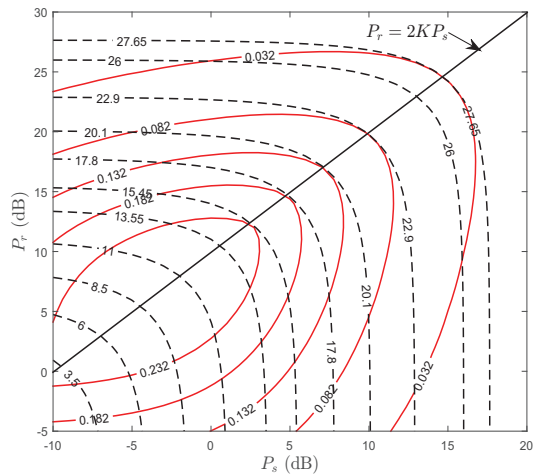}
\caption{Contour plot of energy efficiency versus $P_s$ and $P_r$, with $N=256$, $\kappa=0.375$, $P_{const}=20$W, $P_p=10$dB, $P_0=1$W and $P_{APS}=0.02$W.}
\label{Fig:Contour}
\end{figure}

Fig.~\ref{Fig:Green_points_line} presents the EE-SE relationship under perfect CSI for different combinations of system parameters $(P_0, P_{const}, P_{APS})$. As revealed in \emph{Corollary}~1, there always exists a single green point and the green point $(\log(\eta_{EE}^*), \eta_{SE}^*)$ does lie in a decreasing straight line, which implies a larger $\eta_{SE}^*$ always leads to a lower $\eta_{EE}^*$ in the optimal EE design.
\begin{figure}[htbp]
\centering
\includegraphics[width=0.4\textwidth, height=0.35\textwidth]{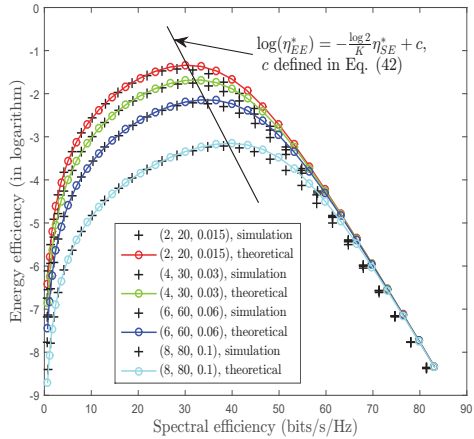}
\caption{EE vs. SE under perfect CSI and different combinations of $(P_0, P_{const}, P_{APS})$, with $N=128$ and $\kappa=0.375$.}
\label{Fig:Green_points_line}
\end{figure}

Fig.~\ref{Fig:EE_optimal_K_different_N} depicts the relationship of system EE versus the RF chain number $L=2K$ under $N=128$, $256$ and $512$, respectively. It is clearly shown that the optimal $L$ falls into the regime $N>\left\lfloor 4L^2/\pi \right\rfloor$ for different numbers of antennas $N$, which validates our statements in Section~\ref{sec:EE}. The optimal RF-chain number is not sensitive to $N$, as the three curves give almost the same optimal choice around $L^*=2K^*=14$, while the EE is dramatically decreased when the relay is equipped with more antennas for a fixed number of RF chains. This is because the SE scales as $\log N$ while the total power consumption scales linearly with $N$.

\begin{figure}[htbp]
\centering
\includegraphics[width = 0.4\textwidth, height=0.35\textwidth]{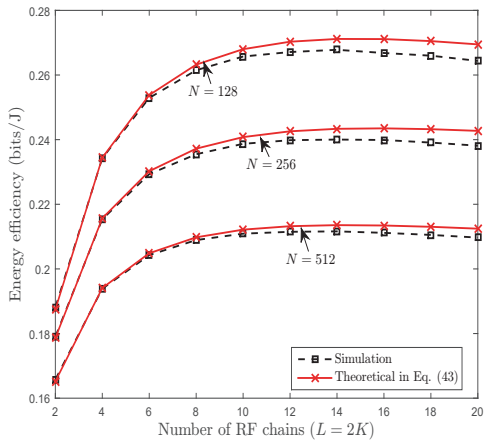}
\caption{EE vs. the number of RF chains with $P_s=P_p=5$ dB, $\kappa=0.375$, $P_{const}=20$W, $P_0=1$W, and $P_{APS}=20$mW.}
\label{Fig:EE_optimal_K_different_N}
\end{figure}

\section{Conclusion}\label{sec:conclusion}

In this paper, we analyzed both the spectral and energy efficiency of a massive MIMO relay network with practical RF-chain constraint. When $N>\left\lfloor 4L^2/\pi \right\rfloor$ is satisfied, for a fixed pilot power, each user and the relay can achieve power saving by scaling down the source and relay transmit power by $1/N$ and $2K/N$, respectively. While if we scale down the transmit power of the pilot and data transmission simultaneously, each user and the relay can only scale down their transmit power by $1/\sqrt{N}$ and $2K/\sqrt{N}$, respectively. In terms of EE, we prove that the maximal EE is always achieved at $P_r =2K P_s$. This condition happens to be also the optimal power splitting strategy for EE maximization under an arbitrary total power constraint. Further given a fixed number of RF-chains, there exists a globally optimal transmit power which yields the best EE performance.

\appendices
\section{Lemma \ref{lemma_uplink}}
\begin{mylemma}\label{lemma_uplink}
For the uplink phase of the two-way massive relay network, the achievable rate $\tilde{R}_{k\rightarrow R}$ behaves as
\begin{align}
\lim\limits_{N\rightarrow\infty}\frac{\tilde{R}_{k\rightarrow R}}{\log_2\left(1+\frac{P_s\pi N\sigma_{k}^2}{4(1+P_s\sum_{j=1}^{2K}\varepsilon_{j}^2)}\right)}=1.
\end{align}
\end{mylemma}
\begin{proof}
In order to derive the asymptotic rate, we first calculate the expectation terms in \eqref{eq:uplinkapprox} one by one in the following.

$\cdot$ Compute $\mathbb{E}\left\{\mathbf{w}_{k}^T\mathbf{F}_r\mathbf{g}_{k}\right\}$: Because $\mathbf{W}_r=[\mathbf{F}_r\hat{\mathbf{G}}]^{-1}$ and $\bm{\mathcal{E}}_G=\mathbf{G}-\hat{\mathbf{G}}$, we have
\begin{align}\label{eq:ZF}
\mathbf{W}_r\mathbf{F}_r\mathbf{G}=\mathbf{W}_r\mathbf{F}_r(\hat{\mathbf{G}}+\bm{\mathcal{E}}_G)=\mathbf{I}_{2K}+\mathbf{W}_r\mathbf{F}_r\bm{\mathcal{E}}_G.
\end{align}
It directly follows
\begin{equation}\label{eq:ZFdiagonal}
\mathbf{w}_{k}^T\mathbf{F}_r\mathbf{g}_{k}=1+\mathbf{w}_{k}^T\mathbf{F}_r\mathbf{\mathcal{E}}_{k},
\end{equation}
where $\mathbf{\mathcal{E}}_{k}$ and $\mathbf{w}_{k}^T\mathbf{F}_r$ are independent, and $\mathbf{\mathcal{E}}_{k}$ is a zero-mean random vector. By taking expectations over $\mathbf{\mathcal{E}}_{k}$ in \eqref{eq:ZFdiagonal}, we have
\begin{equation}\label{eq:desired1}
\mathbb{E}\left\{\mathbf{w}_{k}^T\mathbf{F}_r\mathbf{g}_{k}\right\}=1.
\end{equation}

$\cdot$ Compute $\mathbb{V}\{\mathbf{w}_{k}^T\mathbf{F}_r\mathbf{g}_{k}\}$: From \eqref{eq:ZFdiagonal} and \eqref{eq:desired1}, it gives
\begin{align}\label{eq:Var1}
\mathbb{V}\{\mathbf{w}_{k}^T\mathbf{F}_r\mathbf{g}_{k}\}=&\mathbb{E}\left\{|\mathbf{w}_{k}^T\mathbf{F}_r\mathbf{\mathcal{E}}_{k}|^2\right\}=\varepsilon^2_{k}\mathbb{E}\left\{\mathbf{w}_{k}^T\mathbf{F}_r\mathbf{F}_r^H\mathbf{w}_{k}^*\right\}.
\end{align}
From \eqref{eq:analogprecoder}, entries of $\mathbf{F}_r$ are i.i.d. variables with zero mean and variance $\frac{1}{N}$. Applying the law of large numbers, we have $\mathbf{F}_r\mathbf{F}_r^H\xrightarrow{a.s.}\mathbf{I}_{2K}$. Then
\begin{equation}\label{eq:variance}
\mathbb{V}\{\mathbf{w}_{k}^T\mathbf{F}_r\mathbf{g}_{k}\} = \varepsilon^2_{k}\mathbb{E}\left\{\|\mathbf{w}_{k}^T\|^2\right\}.
\end{equation}

$\cdot$ Compute $\text{MP}_{k}$: According to \eqref{eq:ZF}, for $k\neq j$, we have
\begin{equation}\label{eq:ZFoffdiagonal}
\mathbf{w}_{k}^T\mathbf{F}_r\mathbf{g}_{j}=\mathbf{w}_{k}^T\mathbf{F}_r\mathbf{\mathcal{E}}_{j}.
\end{equation}
Following the similar method for computing $\mathbb{V}\{\mathbf{w}_{k}^T\mathbf{F}_r\mathbf{g}_{k}\}$, we get
\begin{align}\label{eq:MPAN}
\text{MP}_{k}=P_s\mathbb{E}\left\{\|\mathbf{w}_{k}^T\|^2\right\}\sum_{j\neq k}^{2K}\varepsilon^2_{j},~~\text{AN}_{k}=\mathbb{E}\left\{\|\mathbf{w}_{k}^T\|^2\right\}.
\end{align}
By using the results in \eqref{eq:desired1}, \eqref{eq:variance} and \eqref{eq:MPAN}, we have rewritten \eqref{eq:uplinkapprox} as:
\begin{align}\label{eq:wk}
\tilde{R}_{k\rightarrow R}=\log_2\left(1+\frac{P_s}{[1+P_s\sum_{j=1}^{2K}\varepsilon_{j}^2]\mathbb{E}\left\{\|\mathbf{w}_{k}^T\|^2\right\}}\right).
\end{align}
Now, the remaining task is to investigate the expression for $\mathbb{E}\left\{\|\mathbf{w}_{k}^T\|^2\right\}$ where $\mathbf{w}_{k}^T$ comes from the digital ZF precoder $\mathbf{W}_r=[\mathbf{F}_r\hat{\mathbf{G}}]^{-1}$. Denote $\mathbf{f}_{k}$ as the $k$-th column of $\mathbf{F}_r$.
Consider the equivalent uplink channel seen from baseband:
\begin{align}\label{eq:equivalentchannel}
\mathbf{H}_{eq}=\mathbf{F}_r\hat{\mathbf{G}}
=[\mathbf{f}_{1}~\mathbf{f}_{2}\cdots\mathbf{f}_{2K-1}~\mathbf{f}_{2K}]^T[\hat{\mathbf{g}}_{1}~\hat{\mathbf{g}}_{2}\cdots\hat{\mathbf{g}}_{2K-1}~\hat{\mathbf{g}}_{2K}].
\end{align}
The diagonal terms are $h_{k,k}=\mathbf{f}^T_{k}\hat{\mathbf{g}}_{k}$ $(k=1,\cdots,2K)$  where $h_{i,j}$ is the $(i,j)$-th element of $\mathbf{H}_{eq}$. Denote the $i$-th element of $\hat{\mathbf{g}}_{k}$ and $\mathbf{f}_{k}$ with $\hat{g}_{i,k}$ and $f_{i,k}$, respectively. Because
$\mathbf{f}_{k}$ is designed as
$$f_{i,k}=\frac{1}{\sqrt{N}}e^{-j\arg(\hat{g}_{i,k})},$$ we have
\begin{equation}
h_{k,k}=\mathbf{f}^T_{k}\hat{\mathbf{g}}_{k}=\frac{1}{\sqrt{N}}\sum\limits_{i=1}^{N}|\hat{g}_{i,k}|.
\end{equation}

Recalling that $\hat{\mathbf{g}}_{k}\sim\mathcal{CN}(\mathbf{0},\sigma^2_{k}\mathbf{I}_N)$, $\{\hat{g}_{i,k}\}$'s are i.i.d. as $\mathcal{CN}(0,\sigma^2_{k})$. Then $|\hat{g}_{i,k}|$ follows the Rayleigh distribution with mean $\frac{\sigma_k \sqrt{\pi}}{2}$ and variance
$(1-\frac{\pi}{4})\sigma_k^2$. Applying the Central Limit Theorem, it indicates that:
\begin{equation}\label{eq:diagonal}
h_{k,k}\sim\mathcal{N}\left(\frac{\sigma_k \sqrt{\pi N}}{2},\bigg(1-\frac{\pi}{4}\bigg)\sigma_k^2\right),~N\rightarrow \infty.
\end{equation}

On the other hand, the law of large numbers indicates that
\begin{equation}\label{eq:diagonal_almostsurely}
\frac{h_{k,k}}{\sqrt{N}}=\frac{\mathbf{f}^T_{k}\hat{\mathbf{g}}_{k}}{\sqrt{N}}=\frac{1}{N}\sum\limits_{i=1}^{N}|\hat{g}_{i,k}|\xrightarrow{a.s.}\mathbb{E}\{|\hat{g}_{i,k}|\}=\frac{\sigma_k \sqrt{\pi}}{2}.
\end{equation}

The off-diagonal terms are $h_{j,k}=\mathbf{f}^T_{j}\hat{\mathbf{g}}_{k}~(j\neq k)$, analyzing its real and imaginary parts respectively using the Central Limit Theorem followed by proving their independence reveals that \cite{liang2014low}
\begin{equation}\label{eq:offdiagonal}
h_{j,k}\sim\mathcal{CN}\left(0,\sigma_k^2\right).
\end{equation}
Now we rewrite $\mathbf{H}_{eq}$ by its diagonal and off-diagonal components as
\begin{equation}\label{eq:Heq_rewrite}
\mathbf{H}_{eq}=\mathbf{D}+\mathbf{A},
\end{equation}
where $\mathbf{D}=\text{diag}(h_{1,1},h_{2,2},\cdots,h_{2K-1,2K-1},h_{2K,2K}),$
and the $(k,j)$-th element of $\mathbf{A}$ is $h_{k,j}~(k\neq j)$ while the diagonal elements are all zeros. Notice that the off-diagonal terms $h_{k,j}$ could be treated negligible compared to the diagonal terms $h_{k,k}$ when $N$ goes infinitely large. Intuitively the uplink ZF precoder $\mathbf{W}_r=\mathbf{H}_{eq}^{-1}$ can be well approximated as $\mathbf{D}^{-1}$ when $N\rightarrow \infty$. In fact, this is proved by applying \emph{Propositions}~\ref{matrix_inversion} and \ref{taylor_expansion} in Appendix~C. By applying \emph{Proposition}~\ref{matrix_inversion}, we have $$\lim\limits_{N\rightarrow\infty}\mathbf{D}(\mathbf{D}+\mathbf{A})^{-1}=\lim\limits_{N\rightarrow\infty}\mathbf{D}\mathbf{W}_r=\mathbf{I}_{2K},$$ which yields $\lim\limits_{N\rightarrow\infty}h_{k,k}\mathbf{w}_k^T=\mathbf{e}_k^T$, where $\mathbf{e}_k^T$ is the $k$-th row of $\mathbf{I}_{2K}$. Therefore, it is calculated that $\lim\limits_{N\rightarrow\infty}\|\mathbf{w}_k\|^2 h_{k,k}^2$=1, or equivalently $\lim\limits_{N\rightarrow\infty}\|\mathbf{w}_k\|^2=\frac{1}{h_{k,k}^2}$ which further gives $$\lim\limits_{N\rightarrow\infty}\frac{\mathbb{E}\{\|\mathbf{w}_{k}\|^2\}}{\mathbb{E}\left\{\frac{1}{h_{k,k}^2}\right\}}=1=\lim\limits_{N\rightarrow\infty}\frac{\mathbb{E}\{\|\mathbf{w}_{k}\|^2\}}{\mathbb{E}\left\{1/(y+\frac{\sigma_k \sqrt{\pi N}}{2})^2\right\}}.$$
 Meanwhile, according to \emph{Proposition 3} in Appendix~C, we have $\lim\limits_{N\rightarrow\infty}\frac{\mathbb{E}\left\{1/(y+\frac{\sigma_k \sqrt{\pi N}}{2})^2\right\}}{4/\sigma_k^2 \pi N}=1,$
where $y\sim\mathcal{N}(0,\sigma^2)$ with $\sigma^2=(1-\frac{\pi}{4})\sigma_k^2$. Hence,
\begin{equation}\label{eq:wkexpression}
\lim\limits_{N\rightarrow\infty}\frac{\mathbb{E}\{\|\mathbf{w}_{k}\|^2\}}{4/\sigma_k^2 \pi N}=1.
\end{equation}
Now we complete the proof by substituting \eqref{eq:wkexpression} into \eqref{eq:wk}.
\end{proof}

\section{Lemma \ref{lemma_downlink}}
\begin{mylemma}\label{lemma_downlink}
For the downlink transmission of the two-way massive relay network, the achievable rate $\tilde{R}_{R\rightarrow k'}$ behaves as
\begin{align}
\lim\limits_{N\rightarrow\infty}\frac{\tilde{R}_{R\rightarrow k'}}{\log_2\left( 1+\frac{P_r\pi N} {4(1+P_r\varepsilon^2_{k'})\sum_{j=1}^{2K}\frac{1}{\sigma^2_{j}}}\right)}=1.
\end{align}
\end{mylemma}
\begin{proof}
Following a similar method for deriving \eqref{eq:wk}, after some mathematical manipulations, \eqref{eq:received2} reduces to
\begin{equation}\label{eq:received2reduce}
\tilde{R}_{R\rightarrow k'}=\log_2\left(1+\frac{P_r}{\frac{1}{\mu^2}+P_r\varepsilon_{k'}^2\sum_{j=1}^{2K}\mathbb{E}\left\{\|\mathbf{v}_{j}^T\|^2\right\}}\right).
\end{equation}
Because $\mathbf{W}_t= \mathbf{W}_r^T \mathbf{P}$ and considering the function of permutation matrix $\mathbf{P}$, it gives $\mathbf{v}_{k}=\mathbf{w}_{k'}$, where $(k,k')$ is a communication pair.
Thus
\begin{equation}\label{eq:mu}
\lim\limits_{N\rightarrow\infty}\frac{\mathbb{E}\left\{\|\mathbf{v}_{j}^T\|^2\right\}}{4/\sigma_{j'}^2\pi N}=\lim\limits_{N\rightarrow\infty}\frac{\mathbb{E}\left\{\|\mathbf{w}_{j'}^T\|^2\right\}}{4/\sigma_{j'}^2\pi N}=1,
\end{equation}
where the last equality follows from \eqref{eq:wkexpression}. Next, we derive the expression for $\mu$. According to \eqref{eq:relaytransmit} and the transmit power constraint, it follows
\begin{align}
\frac{1}{\mu^2}&=\mathbb{E}\{\text{Tr}[\mathbf{F}_t\mathbf{W}_t\mathbf{x}[t-d]{\mathbf{x}[t-d]}^H{\mathbf{W}_t}^H{\mathbf{F}_t}^H]\}\nonumber\\
&=\mathbb{E}\{\text{Tr}[{\mathbf{W}_t}^H{\mathbf{F}_t}^H\mathbf{F}_t\mathbf{W}_t]\}.
\end{align}
Similar to $\mathbf{F}_r$, due to the law of large numbers, we have ${\mathbf{F}_t}^H\mathbf{F}_t\xrightarrow{a.s.}\mathbf{I}_{2K}$ as $N\rightarrow\infty$.
Therefore
\begin{align}
\lim\limits_{N\rightarrow\infty}\frac{1/\mu^2}{\mathbb{E}\{\text{Tr}[{\mathbf{W}_t}^H\mathbf{W}_t]\}}=1&=\lim\limits_{N\rightarrow\infty}\frac{1/\mu^2}{\mathbb{E}\{\|\mathbf{W}_r\|_F^2\}}\nonumber\\
&=\lim\limits_{N\rightarrow\infty}\frac{1/\mu^2}{\sum_{j=1}^{2K}\|\mathbf{w}_j\|^2}.\nonumber
\end{align}
On the other hand, $\lim\limits_{N\rightarrow\infty}\frac{\sum_{j=1}^{2K}\|\mathbf{w}_j\|^2}{\sum_{j=1}^{2K}4/\sigma_{j}^2\pi N}=1.$ Therefore,
\begin{equation}\label{eq:vj}
\lim\limits_{N\rightarrow\infty}\frac{1/\mu^2}{\sum_{j=1}^{2K}4/\sigma_{j}^2\pi N}=1.
\end{equation}
Substituting \eqref{eq:mu} and \eqref{eq:vj} into \eqref{eq:received2reduce}, we finally arrive at \emph{Lemma}~\ref{lemma_downlink}.
\end{proof}

\section{Propositions \ref{matrix_inversion} and \ref{taylor_expansion}}
\begin{mypro}\label{matrix_inversion}
If the number of relay antennas $N$ and the RF chain number $L$ satisfy $N>\left\lfloor 4L^2/\pi \right\rfloor$, $\mathbf{D}(\mathbf{D}+\mathbf{A})^{-1}$ converges to $\mathbf{I}_{2K}$ in probability as $N\rightarrow\infty$, where $\mathbf{D}$ and $\mathbf{A}$ are defined in \eqref{eq:Heq_rewrite}.
\end{mypro}
\begin{proof}
Since $\mathbf{D}$ is a diagonal matrix with its diagonal elements as $h_{k,k}\sim\mathcal{N}\bigg(\frac{\sigma_k \sqrt{\pi N}}{2},(1-\frac{\pi}{4})\sigma_k^2\bigg)$ when $N\rightarrow \infty$, it is common and reasonable to assume that $\mathbf{D}$ is invertible and the probability of non-invertible $\mathbf{D}$ is in principle zero. Thus we can decompose the matrix $(\mathbf{D}+\mathbf{A})^{-1}$ as follows
\begin{align}\label{eq:inversion}
(\mathbf{D}+\mathbf{A})^{-1}&=\mathbf{D}^{-1}(\mathbf{I}+\mathbf{A}\mathbf{D}^{-1})^{-1}\nonumber\\
&=\mathbf{D}^{-1}+\sum_{k=1}^{\infty}(-1)^k\mathbf{D}^{-1}(\mathbf{A}\mathbf{D}^{-1})^k,
\end{align}
where the last equality uses the well-known matrix decomposition
$
(\mathbf{I}-\mathbf{B})^{-1}=\mathbf{I}+\sum_{k=1}^{\infty}\mathbf{B}^k$, in which the condition for convergence is $\|\mathbf{B}\|_F^2<1$.
Accordingly, it is necessary to check the condition of $X\triangleq\|\mathbf{A}\mathbf{D}^{-1}\|_F^2<1$ in order to guarantee the convergence of series summation in \eqref{eq:inversion}. In the following, we will prove that $\lim\limits_{N\rightarrow\infty}\Pr(X<1)=1$ and hence the convergence of \eqref{eq:inversion} holds for large $N$ in probability.

Recalling that $\mathbf{H}_{eq}=\mathbf{D}+$$\mathbf{A}$ where $\mathbf{D}$ and $\mathbf{A}$ are respectively the diagonal and off-diagonal components of $\mathbf{H}_{eq}$. We have
\begin{equation}\label{eq:convergence}
X=\|\mathbf{A}\mathbf{D}^{-1}\|_F^2
=\sum_{k=1}^{2K}\sum_{j\neq k}^{2K}\frac{|h_{j,k}|^2}{h_{k,k}^2},
\end{equation}
where $h_{i,j}$ is the $(i,j)$-th element of $\mathbf{H}_{eq}$ defined in \eqref{eq:equivalentchannel}.

From \eqref{eq:convergence} and incorporating \eqref{eq:diagonal_almostsurely}, it yields
\begin{equation}\label{eq:norm_almost_sure}
NX=\sum_{k=1}^{2K}\sum_{j\neq k}^{2K}\frac{|h_{j,k}|^2}{(h_{k,k}/\sqrt{N})^2}\xrightarrow{a.s.}\frac{4}{\pi}\sum_{k=1}^{2K}\sum_{j\neq k}^{2K}\frac{|h_{j,k}|^2}{\sigma_k^2}.
\end{equation}
Subsequently, we use \eqref{eq:norm_almost_sure} to evaluate $NX$ in the limit of an infinite $N$.
Consequently,
\begin{equation}
\lim\limits_{N\rightarrow\infty}N\mathbb{E}\{X\}=\frac{4}{\pi}\sum_{k=1}^{2K}\sum_{j\neq k}^{2K}\frac{\mathbb{E}\{|h_{j,k}|^2\}}{\sigma_k^2}=\frac{8K(2K-1)}{\pi}.
\end{equation}
From \cite[Theorem 11.4.3]{Larsen2006statistics}, the variance of a linear combination follows
{\small
\begin{equation}\label{eq:sum_variance}
\mathbb{V}\bigg\{\sum\limits_{m=1}^{M}a_m X_m\bigg\}=\sum\limits_{m=1}^{M}a_m^2\mathbb{V}\{X_m\}+\sum\limits_{n\neq m}^{M}a_n a_m \text{Cov}\{X_n,X_m\},
\end{equation}
}
where $\text{Cov}\{X_n,X_m\}=\mathbb{E}\{X_n X_m\}-\mathbb{E}\{X_n\}\mathbb{E}\{X_m\}$ is the covariance of $X_n$ and $X_m$.
According to \eqref{eq:sum_variance}, we evaluate the variance of \eqref{eq:norm_almost_sure}
\begin{align}\label{eq:variance_expansion}
\lim\limits_{N\rightarrow\infty}N^2\mathbb{V}\{X\}&=\frac{16}{\pi^2}\sum\limits_{k=1}^{2K}\sum\limits_{j\neq k}^{2K}\frac{\mathbb{V}\{|h_{j,k}|^2\}}{\sigma_k^4}\nonumber\\
&+\frac{16}{\pi^2}\sum\limits_{j_1\neq j_2\neq k_1\neq k_2}^{2K}\frac{\text{Cov}\{|h_{j_1,k_1}|^2,|h_{j_2,k_2}|^2\}}{\sigma_{k_1}^2\sigma_{k_2}^2}\nonumber\\
&+\frac{16}{\pi^2}\sum\limits_{j_1\neq j_2\neq k}^{2K}\frac{\text{Cov}\{|h_{j1,k}|^2,|h_{j2,k}|^2\}}{\sigma_k^4}\nonumber\\
&+\frac{16}{\pi^2}\sum\limits_{j\neq k_1\neq k_2}^{2K}\frac{\text{Cov}\{|h_{j,k1}|^2,|h_{j,k2}|^2\}}{\sigma_{k_1}^2\sigma_{k_2}^2},
\end{align}
where a sequence of ``$\neq$'' under the summation operation means that any two of the indexes are not equal.
Since $h_{j,k}\sim\mathcal{CN}\left(0,\sigma_k^2\right)$, we have $$\mathbb{V}\{|h_{j,k}|^2\}=\mathbb{E}\{|h_{j,k}|^4\}-\mathbb{E}^2\{|h_{j,k}|^2\}=\sigma_k^4.$$Next, the covariance terms are derived separately for three cases.

\noindent 1) For $j_1\neq j_2\neq k_1\neq k_2$, it is direct to know that $h_{j_1,k_1}$ and $h_{j_2,k_2}$ are independent, thus $\text{Cov}\{|h_{j_1,k_1}|^2,|h_{j_2,k_2}|^2\}=0$.\\
\noindent 2) For $j_1\neq j_2\neq k$, given $\hat{\mathbf{g}}_k$ fixed, $h_{j_1,k}$ and $h_{j_1,k}$ are independent, thus
\begin{equation}
\mathbb{E}\{|h_{j_1,k}|^2|h_{j_2,k}|^2|\hat{\mathbf{g}}_k\}=\mathbb{E}\{|h_{j_1,k}|^2|\hat{\mathbf{g}}_k\}\mathbb{E}\{|h_{j_2,k}|^2|\hat{\mathbf{g}}_k\}.
\end{equation}
Meanwhile,
\begin{align}
\mathbb{E}\{|h_{j_1,k}|^2|\hat{\mathbf{g}}_k\}=\mathbb{E}\{\mathbf{f}^T_{j_1}\hat{\mathbf{g}}_{k}\hat{\mathbf{g}}^H_{k}\mathbf{f}_{j_1}^*|\hat{\mathbf{g}}_k\}&=\text{Tr}(\hat{\mathbf{g}}_k\hat{\mathbf{g}}_k^H\mathbb{E}\{\mathbf{f}_{j_1}^*\mathbf{f}^T_{j_1}\}|\hat{\mathbf{g}}_k)\nonumber\\
&=\|\hat{\mathbf{g}}_k\|^2/N.
\end{align}
Similarly, $\mathbb{E}\{|h_{j_2,k}|^2|\hat{\mathbf{g}}_k\}=\|\hat{\mathbf{g}}_k\|^2/N$, which yields $\mathbb{E}\{|h_{j_1,k}|^2|h_{j_2,k}|^2|\hat{\mathbf{g}}_k\}=\|\hat{\mathbf{g}}_k\|^4/N^2$. Further taking expectations over $\hat{\mathbf{g}}_k$ yields
\begin{equation}
\mathbb{E}\{|h_{j_1,k}|^2|h_{j_2,k}|^2\}=\mathbb{E}\{\|\hat{\mathbf{g}}_k\|^4\}/N^2=(1+\frac{1}{N})\sigma_k^4,
\end{equation}
where the last equality is due to the fact that $\|\hat{\mathbf{g}}_k\|^2$ follows a Gamma distribution $\Gamma(N,\sigma_k^2)$, and $\mathbb{E}\{\|\hat{\mathbf{g}}_k\|^4\}=N(N+1)\sigma_k^4$.
Consequently,
\begin{align}
&\lim\limits_{N\rightarrow\infty}\text{Cov}\{|h_{j_1,k}|^2,|h_{j_2,k}|^2\}\nonumber\\
=&\lim\limits_{N\rightarrow\infty}\mathbb{E}\{|h_{j_1,k}|^2|h_{j_2,k}|^2\}-\mathbb{E}\{|h_{j_1,k}|^2\}\mathbb{E}\{|h_{j_2,k}|^2\}\nonumber\\
=&\lim\limits_{N\rightarrow\infty}\frac{\sigma_k^4}{N}=0.
\end{align}
3) For $j\neq k_1\neq k_2$, by first deriving the conditional covariance with given $\mathbf{f}_j$ and then taking expectations over $\mathbf{f}_j$, we can obtain
\begin{equation}
\lim\limits_{N\rightarrow\infty}\text{Cov}\{|h_{j,k1}|^2,|h_{j,k2}|\}=0.
\end{equation}
Based on the above results, \eqref{eq:variance_expansion} reduces to
\begin{equation}
\lim\limits_{N\rightarrow\infty}N^2\mathbb{V}\{X\}=\frac{32K(2K-1)}{\pi^2}
\end{equation}
Applying the Chebyshev inequality, we obtain for any $0<\epsilon<1$,
$$\lim\limits_{N\rightarrow\infty}\Pr(|X-\mathbb{E}\{X\}|\leq \epsilon)\geq \lim\limits_{N\rightarrow\infty}1-\frac{\mathbb{V}\{X\}}{\epsilon^2}=1,$$
which implies that $X=\|\mathbf{A}\mathbf{D}^{-1}\|_F^2$ converges to $\mathbb{E}\{X\}=\frac{8K(2K-1)}{\pi N}$ in probability for large $N$.

Now that by incorporating the convergence condition of $\|\mathbf{A}\mathbf{D}^{-1}\|_F^2<1$ for \eqref{eq:inversion}, we can conclude that the condition $\|\mathbf{A}\mathbf{D}^{-1}\|_F^2<1$ equivalently becomes $\mathbb{E}\{X\}<1$ for large $N$, which yields
\begin{equation}\label{eq:N_K_relationship}
N>\left\lfloor\frac{8K(2K-1)}{\pi}\right\rfloor.
\end{equation}
Define $\delta\triangleq\|\mathbf{D}(\mathbf{D}+\mathbf{A})^{-1}-\mathbf{I}_{2K}\|_F$, with the convergence proved, we can apply the expansion in \eqref{eq:inversion} and then $\delta=\|\sum_{k=1}^{\infty}(-1)^k(\mathbf{A}\mathbf{D}^{-1})^k\|_F$. Further, by successively using the triangle inequality of Frobenius norm and the fact that $\|\mathbf{XY}\|_F\leq\|\mathbf{X}\|_F\|\mathbf{Y}\|_F$ \cite{hornmatrix}, it follows
\begin{align}
\lim\limits_{N\rightarrow\infty}\delta&\leq\lim\limits_{N\rightarrow\infty}\sum_{k=1}^{\infty}\|(\mathbf{A}\mathbf{D}^{-1})^k\|_F\nonumber\\
&\leq\lim\limits_{N\rightarrow\infty}\sum_{k=1}^{\infty}\|\mathbf{A}\mathbf{D}^{-1}\|_F^k\nonumber\\
&= \lim\limits_{N\rightarrow\infty}\frac{\|\mathbf{A}\mathbf{D}^{-1}\|_F}{1-\|\mathbf{A}\mathbf{D}^{-1}\|_F}.
\end{align}
When $N\rightarrow\infty$, $\|\mathbf{A}\mathbf{D}^{-1}\|_F^2$ converges to $\frac{8K(2K-1)}{\pi N}$ in probability and hence $\delta$ converges to zero. Therefore, we conclude that $\mathbf{D}(\mathbf{D}+\mathbf{A})^{-1}\rightarrow\mathbf{I}_{2K}$ in probability.
\end{proof}
\begin{mypro}\label{taylor_expansion}
For a random variable $y\sim\mathcal{N}(0,\sigma^2)$, the following equality holds:
\begin{equation}\label{eq:proposition2}
\lim\limits_{a\rightarrow\infty} \mathbb{E} \left\{ \frac{a^2}{(y+a)^2}\right\}=1.
\end{equation}
\end{mypro}
\begin{proof}
For $y\sim\mathcal{N}(0,\sigma^2)$, the Chebyshev inequality gives $
\Pr\left(\left|\frac{y}{a}\right|<1\right)\geq 1-\frac{\sigma^2}{a^2}$,
which further yields $\lim\limits_{a\rightarrow\infty}\Pr\left(\left|\frac{y}{a}\right|<1\right)=1$.
Therefore, we can have the following expansion
{\small
\begin{equation}\label{eq:expansion}
\lim\limits_{a\rightarrow\infty}\frac{a^2}{(y+a)^2}=\lim\limits_{a\rightarrow\infty}\frac{1}{(1+\frac{y}{a})^2}=\lim\limits_{a\rightarrow\infty}\sum\limits_{k=0}^{\infty}(-1)^k(k+1)\left(\frac{y}{a}\right)^k,
\end{equation}
}
which uses the Taylor expansion $$\frac{1}{(1+x)^2}=\sum\limits_{k=0}^{\infty}(-1)^k(k+1)x^k,~|x|<1.$$
By exploiting the results on central moments of Gaussian random variables, i.e.,
\begin{align}
\mathbb{E}\{y^{2k-1}\}=0, \ \mathbb{E}\{y^{2k}\}=\sigma^{2k}(2k-1)!!,~k=1,2,\cdots,
\end{align}
the expectation of \eqref{eq:expansion} is calculated as
$$\lim\limits_{a\rightarrow\infty}\mathbb{E}\left\{\frac{a^2}{(y+a)^2}\right\}
=1+\lim\limits_{a\rightarrow\infty}\sum\limits_{k=1}^{\infty}\frac{\sigma^{2k}(2k+1)!!}{a^{2k}}=1,$$
which gives the desired result.
\end{proof}
\section{Lemma \ref{lemma_power}}\label{proof_lemma_power}
\begin{mylemma}\label{lemma_power}
Given $K$ and $N$ fixed, the optimal EE is always achieved at $2KP_s=P_r$ for $P_s>0$ and $P_r>0$.
\end{mylemma}
\begin{proof}
Consider $\eta_{EE}$ as a function of variables $P_s$ and $P_r$, denoted as $\eta_{EE}(P_s,P_r)$. The feasible region of $(P_s,P_r)$ can be divided into three subregions: $\mathcal{S}_0=\{(P_s,P_r)|2KP_s=P_r>0\}$, $\mathcal{S}_1=\{(P_s,P_r)|0<2KP_s<P_r\}$ and $\mathcal{S}_2=\{(P_s,P_r)|2KP_s>P_r>0\}$.

\noindent 1)~~For any $(P_s,P_r)\in \mathcal{S}_1$, i.e., $2KP_s<P_r$ which gives $\Delta=\frac{2KP_s}{1+2KP_s\varepsilon_0^2}$ in \eqref{eq:EE1}, we have
{\small
\begin{align}\label{eq:S1}
 ~~\eta_{EE}(P_s,P_r)=\frac{K\log_2\bigg(1+\frac{\pi N\sigma_0^2}{8K}\Delta\bigg)}{\frac{1}{2\kappa}(2KP_s+P_r)+2KP_0+P_{const}+2KNP_{APS}}.
\end{align}
}
Compare \eqref{eq:S1} with the special case of $P_r=2KP_s$, i.e.,
\begin{align}
&\eta_{EE}(P_s,2KP_s)\nonumber\\
=&\frac{K\log_2\bigg(1+\frac{\pi N\sigma_0^2}{8K}\Delta\bigg)}{\frac{1}{2\kappa}(2KP_s+2KP_s)+2KP_0+P_{const}+2KNP_{APS}}.
\end{align}
Due to $2KP_s<P_r$, it directly follows $\eta_{EE}(P_s,2KP_s)>\eta_{EE}(P_s,P_r)$. Therefore, for any $(P_s,P_r)\in \mathcal{S}_1$, there exists $(P_s,2KP_s)\in \mathcal{S}_0$ such that $\eta_{EE}(P_s,2KP_s)>\eta_{EE}(P_s,P_r)$.

\noindent 2)~~For any $(P_s,P_r)\in \mathcal{S}_2$, i.e., $2KP_s>P_r$, following similar procedures as in 1), it is easily checked that for any $2KP_s>P_r$, it holds $$\eta_{EE}(P_r/{2K},P_r)>\eta_{EE}(P_s,P_r).$$
Thus, for any point $(P_s,P_r)\in \mathcal{S}_2$, there exists $(P_r/{2K},P_r)\in \mathcal{S}_0$ yielding $\eta_{EE}(P_r/{2K},P_r)>\eta_{EE}(P_s,P_r)$. Consequently, based on 1) and 2), we conclude that the maximal $\eta_{EE}$ must be achieved in $\mathcal{S}_0$, which implies that the optimal $(P_s,P_r)$ which gives the maximal EE always satisfies $P_r=2KP_s$.

In fact, the condition of $2KP_s=P_r$ is also the optimal power allocation strategy in terms of EE under an arbitrary total power constraint as $2KP_s+P_r=P_T$. This can be easily verified by checking the following two facts.\\
i)~~If $2KP_s\geq P_r$, i.e., $0<P_r\leq\frac{P_T}{2}$, $$\eta_{SE}=K\log2\bigg(1+\frac{P_r\pi N\sigma_0^2}{8K(1+P_r\epsilon_0^2)}\bigg)$$ is an increasing function of $P_r$. While when $P_r+2KP_s=P_T$, the total power consumption is constant, thus the maximal SE directly yields the maximal EE. Hence, the maximal EE is achieved when $P_r=\frac{P_T}{2}$, equivalently $P_r=2KP_s$.\\
ii)~~If $2KP_s\leq P_r$, i.e., $0<2KP_s\leq\frac{P_T}{2}$, it is checked that
$\eta_{SE}$ also increases with $P_s$. Therefore, the maximum EE is achieved when $2KP_s=P_r=\frac{P_T}{2}$.
\end{proof}

\section{Proposition \ref{Ps_star}}\label{optimal_Ps}
\begin{mypro}\label{Ps_star}
For fixed $K$ and $N$, $\eta_{EE}$ in \eqref{eq:EE2} is quasi-concave with respect to $P_s$.
\end{mypro}
\begin{proof}
Define $f_1(P_s)=K\log_2\bigg(1+\frac{a_0 P_s}{a_1 P_S+4}\bigg)$, and let $g_1(P_s)=\frac{2KP_s}{\kappa}+P_{c}$. Then \eqref{eq:EE2} can be rewritten as
\begin{equation}
\eta_{EE}=\frac{f_1(P_s)}{g_1(P_s)}.
\end{equation}
The second order derivative of $f_1(P_s)$ is checked as
\begin{align}
f''_1(P_s)&=-\frac{4a_0 K}{\ln2}\frac{(a_0+a_1)(a_1 P_s+4)+a_1[(a_0+a_1)P_s+4]}{[(a_0+a_1)P_s+4]^2(a_1 P_s+4)^2}\nonumber\\
&<0,
\end{align}
which means $f_1(P_s)$ is concave with regard to $P_s$. On the other hand, $g_1(P_s)$ is linear, and hence it is convex with respect to $P_s$. Further since both $f_1(P_s)$ and $g_1(P_s)$ are differentiable and from \cite{boydconvex}, $\eta_{EE}$ is quasi-concave w.r.t. to $P_s$.
\end{proof}

\end{document}